\documentclass[pra,floatfix,twocolumn,amsfonts,amsmath,amssymb,superscriptaddress,nofootinbib]{revtex4-1}
\pdfoutput=1

\usepackage{amsthm}
\usepackage{mathtools}

\usepackage{bbm}
\usepackage{dsfont}
\usepackage[T1]{fontenc}
\usepackage[utf8]{inputenc}
\usepackage[english]{babel}
\usepackage{lmodern}

\usepackage[colorlinks]{hyperref}

\DeclareMathOperator{\tr}{tr}

\newcommand{\ket}[1] {| #1 \rangle}

\newcommand{\proj}[1]{ | #1 \rangle\!\langle #1 |}

\newcommand{\id}{\mathds{1}}

\newcommand{\Ket}[1] {| #1 \rangle\!\rangle}

\newcommand{\KetBra}[2]{ | #1 \rangle\!\rangle\!\langle\!\langle #2 |}

\newcommand{\Hi}{\mathcal{H}}
\newcommand{\LO}{\mathcal{L}}

\newcommand{\ba}{\begin{eqnarray}}
\newcommand{\ea}{\end{eqnarray}}
\newcommand{\be}{\begin{equation}}
\newcommand{\ee}{\end{equation}}

\newtheorem{thm}{Theorem}

\newtheorem{lem}[thm]{Lemma}
\newtheorem*{lem*}{Lemma}
\newtheorem{definition}[thm]{Definition}

\def\pr#1#2{P_{#1}(a_{#2}|x_{#2})} 
\def\prv#1#2{P_{#1}(\vec{a}_{#2}|\vec{x}_{#2})}

\hypersetup{pdftitle={Multipartite Causal Correlations: Polytopes and Inequalities}}

\begin{document}

\title{Multipartite Causal Correlations: Polytopes and Inequalities}

\author{Alastair A. Abbott}
\affiliation{Institut Néel, CNRS and Universit\'e Grenoble Alpes, 38042 Grenoble Cedex 9, France}
\author{Christina Giarmatzi}
\affiliation{Centre for Engineered Quantum Systems, School of Mathematics and Physics, The University of Queensland, St. Lucia, QLD 4072, Australia}
\affiliation{Centre for Quantum Computation and Communication Technology, School of Mathematics and Physics, The University of Queensland, St. Lucia,   QLD 4072, Australia}
\author{Fabio Costa}
\affiliation{Centre for Engineered Quantum Systems, School of Mathematics and Physics, The University of Queensland, St. Lucia, QLD 4072, Australia}
\author{Cyril Branciard}
\affiliation{Institut Néel, CNRS and Universit\'e Grenoble Alpes, 38042 Grenoble Cedex 9, France}

\date{\today}

\begin{abstract}
We consider the most general correlations that can be obtained by a group of parties whose causal relations are well-defined, although possibly probabilistic and dependent on past parties' operations. We show that, for any fixed number of parties and inputs and outputs for each party, the set of such correlations forms a convex polytope, whose vertices correspond to deterministic strategies, and whose (nontrivial) facets define so-called causal inequalities. We completely characterize the simplest tripartite polytope in terms of its facet inequalities, propose generalizations of some inequalities to scenarios with more parties, and show that our tripartite inequalities can be violated within the process matrix formalism, where quantum mechanics is locally valid but no global causal structure is assumed.
\end{abstract}

\maketitle

\section{Introduction}

One of the most surprising features of quantum mechanics is that it generates correlations that cannot be obtained with classical systems. The most studied scenario involves parties that cannot communicate with each other. As first proved by Bell~\cite{bell64}, two such parties sharing entangled quantum states can generate nonsignaling correlations that are not achievable with classical resources.
The no-signaling constraint corresponds to a particular \emph{causal structure}, where all correlations are due to a common cause. There is a growing interest in studying and characterizing more general causal structures~\cite{BranciardBilocal2012, FritzBeyond2012, ChavesCausal2014, fritzbeyond2015} and understanding when the corresponding correlations can be produced with classical or quantum systems~\cite{Henson2014, pienaar2016causal}.

But what are the most general correlations achievable in \emph{any} causal structure? And can quantum mechanics generate even more general correlations not compatible with any definite causal structure? This question was considered in~\cite{oreshkov12}, where a framework was developed that assumes the validity of quantum mechanics in local laboratories, with no assumptions about the causal structure in which the laboratories are embedded. It was found that such a framework allows for correlations that can violate \emph{causal inequalities}, constraints that are necessarily satisfied by correlations generated in any definite causal order. However, no clear physical interpretation was found for such \emph{`noncausal correlations'}. 

A physical process in which operations are performed `in a superposition' of causal orders was first proposed in~\cite{Chiribella:2013aa}. This process---the \emph{quantum switch}---can in principle provide advantages for computation~\cite{araujo14} and communication~\cite{feixquantum2015, guerin16}, and a first experimental proof of principle has been recently demonstrated~\cite{Procopio:2015aa}. However, the quantum switch requires device-dependent tests to detect its `causal nonseparability'~\cite{araujo15} (i.e.,\ the lack of definite causal order) and it cannot be used to violate any causal inequality~\cite{araujo15, oreshkov15}.
Interestingly, the quantum switch requires the coordinated action of three parties to detect causal nonseparability, while none of the known bipartite examples of causally nonseparable processes seem to have physical interpretations. This motivates a systematic study of multipartite scenarios: if there is a process in nature that can violate causal inequalities, it may indeed require more than two parties.

Here we consider general scenarios involving a finite but arbitrary number of parties with finite numbers of possible inputs and outputs, and prove that \emph{`causal correlations'}---those that can be generated in a well-defined causal structure, be it fixed or dynamical, deterministic or probabilistic---form a convex polytope whose vertices correspond to deterministic strategies. That is, any causal correlation can be expressed as a probabilistic mixture of deterministic causal strategies. These are strategies where the output of each party is a deterministic function of its own input and of the inputs of parties in its past, and where the causal relations amongst a set of parties are functions of the inputs of the parties in the past of that set.
We further completely characterize the simplest nontrivial polytope for three parties in terms of its facets, which define causal inequalities. We interpret some of these inequalities in terms of device-independent `causal games'~\cite{oreshkov12}, for which the probability of success has a nontrivial upper bound whenever the parties are constrained by a definite causal order. We also generalize some of these inequalities to the $N$-partite case. Finally, we show that all of these nontrivial tripartite inequalities can be violated within the process matrix formalism of Ref.~\cite{oreshkov12}.

\section{Multipartite causal correlations}

\subsection{Scenarios and notation}

In this paper we will consider situations where a finite number $N \ge 1$ of parties $A_k$ each receive an input $x_k$ from some finite set (which can in principle be different for each party), and generate an output $a_k$ that also belongs to some finite set (and which may also differ for each input). The fixed number of parties, of possible inputs for each party, and of possible outputs for each input, define together a `scenario'; throughout the paper we will always (often implicitly) assume that such a scenario is fixed. We define the vectors of inputs and outputs $\vec x = (x_1, \ldots, x_N)$ and $\vec a = (a_1, \ldots, a_N)$. The correlation established by the $N$ parties in a given scenario is then described by the conditional probability distribution $P(\vec a|\vec x)$.

For any (nonempty) subset ${\cal K} = \{k_1, \ldots, k_{|{\cal K}|}\}$ of ${\cal N} := \{1,\ldots,N\}$ with $|\cal K|$ elements, we shall denote by $\vec x_{\cal K} = (x_{k_1}, \ldots, x_{k_{|{\cal K}|}})$ and $\vec a_{\cal K} = (a_{k_1}, \ldots, a_{k_{|{\cal K}|}})$ the lists of inputs and outputs for the parties in ${\cal K}$. This will in particular allow us to consider marginal correlations and write, for instance (noting that $\vec a = (\vec a_{\cal K}, \vec a_{{\cal N} \backslash {\cal K}})$, up to a reordering of the parties), $P(\vec a_{\cal K}|\vec x)  = \sum_{\vec a_{{\cal N} \backslash {\cal K}}} \prv{}{}$. For ease of notation, a singleton ${\cal K} = \{k\}$ will simply be written $k$, and the vectors of inputs and outputs corresponding to the $(N{-}1)$ parties in ${\cal N} \backslash k$ (obtained by removing just one party, $A_k$) will simply be denoted by $\vec x_{\backslash k} = (x_1, \ldots, x_{k-1}, x_{k+1}, \ldots, x_N)$ and $\vec a_{\backslash k} = (a_1, \ldots, a_{k-1}, a_{k+1}, \ldots, a_N)$.

\subsection{Defining multipartite causal correlations}
\label{subsec:definition_causal_correlations}

We wish to investigate here the correlations that can be established in a scenario where each party's events---namely, the choice of an input and the generation of an output---happen within a well-defined causal structure, with well-defined causal relations between the parties.

The case of two parties is rather clear. The only possible causal relations are that $A_1$ causally precedes $A_2$---a case that we denote by $A_1 \prec A_2$, and which implies that $A_1$'s marginal probability distribution or `response function' should not depend on $A_2$'s input: $P(a_1|x_1,x_2) = P(a_1|x_1)$---or vice versa, $A_2 \prec A_1$, where $P(a_2|x_1,x_2) = P(a_2|x_2)$. The case where $A_1$ and $A_2$ are causally independent can be included in either $A_1 \prec A_2$ or $A_2 \prec A_1$, as it is compatible with both. As originally considered in Ref.~\cite{oreshkov12}, one may also allow for situations where the causal order is not fixed, but chosen probabilistically. A bipartite correlation that is compatible with $A_1 \prec A_2$, or $A_2 \prec A_1$, or a probabilistic mixture of the two is said to be \emph{`causal'}~\cite{oreshkov12,araujo15,brukner14,oreshkov15,branciard16}.

Moving now to three or more parties, more complex possibilities arise. Indeed, the action of the first party could control the causal relations of the following parties, in perfect agreement with the idea of a well-defined causal structure. For instance, if a party $A_1$ is first, they could decide to set $A_2$ before $A_3$ if their input is $0$, or $A_3$ before $A_2$ if it is $1$ (or they could choose the causal order between $A_2$ and $A_3$ as the result of a coin toss, where the coin's bias depends on the input). We should thus allow for such \emph{dynamical} causal orders~\cite{hardy2005probability,oreshkov15} (sometimes also referred to as \emph{adaptive} causal orders~\cite{baumeler13}) to establish the correlations we are interested in.

In any case, even allowing for dynamical causal orders, the compatibility with a definite causal structure will always require that one party acts first; which party this is can be chosen probabilistically, as in the bipartite case. The response function of that first party (say $A_k$) should then not depend on the other parties' inputs: $P(a_k|\vec x) = P(a_k|x_k)$. The action of that party would then determine the causal structure of the following parties, so that the correlation shared by the latter, conditioned on the input and output of the former, should also be compatible with a definite causal structure.
These observations lead us to introduce the following inductive definition for what we shall call \emph{`causal correlations'}:

\begin{definition}[Multipartite causal correlations]
\label{def:causal_correlations}
$ $

\begin{itemize}

\item For $N=1$, any valid probability distribution $P(a_1|x_1)$ is \emph{causal};

\item For $N \ge 2$, an $N$-partite correlation is \emph{causal} if and only if it can be decomposed in the form
\ba
P(\vec a|\vec x) = \sum_{k \in {\cal N}} \ q_k \ P_k(a_k|x_k) \ P_{k,x_k,a_k}(\vec a_{\backslash k}|\vec x_{\backslash k}) \,, \label{def:causal_correlation}
\ea
with $q_k \ge 0$ for each $k$, $\sum_k q_k = 1$, where (for each $k$) $P_k(a_k|x_k)$ is a single-party (and hence causal) probability distribution and (for each $k, x_k, a_k$) $P_{k,x_k,a_k}(\vec a_{\backslash k}|\vec x_{\backslash k})$ is a causal $(N{-}1)$-partite correlation.

\end{itemize}
\end{definition}

It is easy to check that this general definition is compatible with that for the bipartite case recalled above.
We note that the multipartite case was first investigated in Ref.~\cite{oreshkov15}. It was shown there that the above definition characterizes precisely the correlations that are compatible with the intuition about causality that (paraphrasing~\cite{oreshkov15}) the choice of input for one party cannot affect the outputs of other parties that acted before it (or which are not causally related, being neither in the past nor in the future of it), nor the causal order between those previous parties and the party in question. This also implies that causal correlations are those for which a classical `hidden variable' exists, whose value determines the causal order between all the parties, and signaling is only possible from parties in the causal past to those in the future according to the given causal order~\cite{oreshkov15}.
These arguments provide further justification to the above choice of definition for multipartite causal correlations.

\subsection{Basic properties of causal correlations}
\label{subsec:a_few_comments}

Let us mention some basic properties of this definition of multipartite causal correlations.
The proofs of the claims below are given in the Appendix.

\subsubsection{Convexity of causal correlations}
\label{subsubsec:convexity}

From the previous discussion it should be clear that any probabilistic mixture of causal correlations (for a given scenario) must also be causal, so that the set of causal correlations is convex. Although this is not immediately evident from the definition of Eq.~\eqref{def:causal_correlation}, it can indeed be shown to be the case.

\subsubsection{Ignoring certain parties}
\label{subsubsec:ignore_parties}

Any marginal correlation, for any subset of parties, of a causal correlation is causal.

More specifically, consider an $N$-partite causal correlation $\prv{}{}$ and a nonempty subset ${\cal K} \subset {\cal N}$.
Then the $|\cal K|$-partite correlation
\ba
P_{\vec x_{{\cal N} \backslash {\cal K}}}(\vec a_{\cal K}|\vec x_{\cal K}) := P(\vec a_{\cal K}|\vec x)  = \sum_{\vec a_{{\cal N} \backslash {\cal K}}} \prv{}{}
\ea
is causal for all $\vec x_{{\cal N} \backslash {\cal K}}$.

The correlation $P_{\vec x_{{\cal N} \backslash {\cal K}}}(\vec a_{\cal K}|\vec x_{\cal K})$ above is still conditioned on the \emph{inputs} of the $N{-}|{\cal K}|$ parties whose outputs are discarded. One can remove this dependence by averaging it out (for a given input distribution): by the convexity of causal correlations, the resulting $|\cal K|$-partite correlation $P(\vec a_{\cal K}|\vec x_{\cal K})$ remains causal. Note, on the other hand, that the correlation $P_{\vec x_{{\cal N} \backslash {\cal K}}, \vec a_{{\cal N} \backslash {\cal K}}}(\vec a_{\cal K}|\vec x_{\cal K})$, conditioned also on the \emph{outputs} of the discarded parties, is \emph{not} necessarily causal: post-selection indeed allows one to turn a causal correlation into a noncausal one.\footnote{To see this, consider for instance a bipartite `causal game' as in~\cite{oreshkov12,branciard16}, and add a third party to whom all inputs and outputs of the first two parties are sent, and who outputs $1$ if and only if the winning conditions for the causal game are met. Postselecting on that output, the bipartite correlation shared by the first two parties clearly wins the game perfectly (which implies that it is noncausal), although it could be established within a well-defined causal structure.}

\subsubsection{Combining causal correlations `one after the other'}

Consider two (nonempty) sets of parties ${\cal K}$ and ${\cal N} \backslash {\cal K}$, and two causal correlations: $P(\vec a_{\cal K}|\vec x_{\cal K})$ for the first set, and $P_{\vec x_{\cal K},\vec a_{\cal K}}(\vec a_{{\cal N} \backslash {\cal K}}|\vec x_{{\cal N} \backslash {\cal K}})$ for the second, which may depend on the inputs and outputs of the first set of parties (the parties in the set $\mathcal{K}$ are thus understood to `act before' those of the set ${\cal N} \backslash {\cal K}$). 
Then the $N$-partite correlation obtained by combining those in the form
\ba
P(\vec a|\vec x) := P(\vec a_{\cal K}|\vec x_{\cal K}) \ P_{\vec x_{\cal K},\vec a_{\cal K}}(\vec a_{{\cal N} \backslash {\cal K}}|\vec x_{{\cal N} \backslash {\cal K}})
\ea
is causal.

\subsubsection{An equivalent characterization of causal correlations}

An equivalent characterization, for the case $N \ge 2$, is that an $N$-partite correlation is causal if and only if it can be decomposed in the form (cf. also Ref.~\cite{oreshkov15})
\ba
P(\vec a|\vec x) = \sum_{\emptyset \varsubsetneq \cal K \varsubsetneq \cal N} q_{\cal K} \ P_{\cal K}(\vec a_{\cal K}|\vec x_{\cal K}) \ P_{{\cal K},\vec x_{\cal K},\vec a_{\cal K}}(\vec a_{{\cal N} \backslash {\cal K}}|\vec x_{{\cal N} \backslash {\cal K}}) \,, \nonumber \\[-3mm] \label{def:causal_correlation_v2}
\ea
where the sum runs over all nonempty strict subsets of ${\cal N} = \{1, \ldots, N\}$ (with $|{\cal K}| < N$ elements), with $q_{\cal K} \ge 0$ for each ${\cal K}$, $\sum_{\cal K} q_{\cal K} = 1$, where (for each ${\cal K}$) $P_{\cal K}(\vec a_{\cal K}|\vec x_{\cal K})$ is a $|{\cal K}|$-partite causal correlation, and (for each ${\cal K}, \vec x_{\cal K}, \vec a_{\cal K}$) $P_{{\cal K},\vec x_{\cal K},\vec a_{\cal K}}(\vec a_{{\cal N} \backslash {\cal K}}|\vec x_{{\cal N} \backslash {\cal K}})$ is an $(N{-}|{\cal K}|)$-partite causal correlation.

\medskip

This equivalent characterization implies in particular that a correlation of the form
\ba
P(\vec a|\vec x) = \sum_{k \in {\cal N}} \ q_k \ P_k(\vec a_{\backslash k}|\vec x_{\backslash k}) \ P_{k,\vec x_{\backslash k},\vec a_{\backslash k}}(a_k|x_k) \,, \label{def:correlation_party_last}
\ea
with $q_k \ge 0$ for each $k$, $\sum_k q_k = 1$, where (for each $k$) $P_k(\vec a_{\backslash k}|\vec x_{\backslash k})$ is a causal $(N{-}1)$-partite correlation and (for each $k, \vec x_{\backslash k},\vec a_{\backslash k}$) $P_{k,\vec x_{\backslash k},\vec a_{\backslash k}}(a_k|x_k)$ is a single-party probability distribution, is causal. (It is indeed obtained from~\eqref{def:causal_correlation_v2} by summing over the subsets ${\cal K} = {\cal N} \backslash k$, and relabelling certain subscripts ${\cal N} \backslash k$ to $k$.)
Compared to Eq.~\eqref{def:causal_correlation}, each term in the above sum distinguishes a given party that `comes \emph{last}', rather than \emph{first}. Note for instance that the correlations obtained from the so-called quantum switch~\cite{Chiribella:2013aa} are precisely of this form, and are hence causal; see Refs.~\cite{oreshkov15,araujo15}.

It should be emphasized however that for $N \ge 3$ the contrary is not true: not all causal correlations are of the form~\eqref{def:correlation_party_last}. Correlations with dynamical causal order, such as the one given below, provide counter-examples.

\subsubsection{Examples: fixed-order, mixtures of fixed orders, and dynamical-order causal correlations}
\label{subsubsec:examples}

Let us finish this section with some examples.

\medskip

The simplest example of a causal correlation one can think of is one that is compatible with a fixed causal order between all the parties that is independent of any party's input and output. For instance, a correlation compatible with the causal order $A_1 \prec A_2 \prec \dots \prec A_N$ can be written:
\ba
\prv{}{} &=& \pr{}{1}\ \pr{x_1,a_1}{2}\ \pr{x_1,x_2,a_1,a_2}{3}\notag \\ 
 && \qquad\qquad\quad \times \, \cdots \, \times\, \pr{\vec{x}_{\backslash N}, \vec{a}_{\backslash N}}{N},
\ea
which clearly satisfies the definition of a causal correlation given by Eq.~\eqref{def:causal_correlation}.

\medskip

Beyond this simplest case, by the convexity of the definition (see Sec.~\ref{subsubsec:convexity} above), any probabilistic mixture of fixed-order causal correlations is causal. For example, if the correlation $P_\sigma$ is compatible with the fixed order $A_{\sigma(1)} \prec A_{\sigma(2)} \prec \dots \prec A_{\sigma(N)}$ and $P_\tau$ is compatible with $A_{\tau(1)} \prec A_{\tau(2)} \prec \dots \prec A_{\tau(N)}$ (where $\sigma$ and $\tau$ are two permutations of $\{1, \ldots, N\}$), then for any $q \in [0,1]$, $P = q P_\sigma + (1{-}q) P_\tau$ is also causal. The interpretation is simply that, with probability $q$, the correlation is compatible with the fixed causal order defined by $\sigma$, while with probability $1{-}q$ it is compatible with $\tau$.

\medskip

For $N \ge 3$, this mixture of fixed-order causal correlations is not yet the most general type of causal correlation. Indeed, as discussed above, the inputs and outputs of the party (or parties) acting first could influence the causal order between the subsequent parties: the causal order can be \emph{dynamical}~\cite{oreshkov15}.
As a concrete example, consider for instance the tripartite scenario with binary inputs $0$, $1$ for all parties, a single fixed output for $A_1$ (which we can therefore ignore) and binary outputs $0$, $1$ for $A_2$ and $A_3$, and the following (deterministic) correlation:
\ba
P(a_2,a_3|x_1,x_2,x_3) &\ = \ & \delta_{x_1,0} \ \delta_{a_2,0} \ \delta_{a_3,x_2} \nonumber \\
&& + \ \delta_{x_1,1} \ \delta_{a_2,x_3} \ \delta_{a_3,0} \,,
\label{dynamical}
\ea
where $\delta$ is the Kronecker delta.
This example can be understood causally as follows (recall the discussion of Subsection~\ref{subsec:definition_causal_correlations}): the party $A_1$ acts first; their input ($0$ or $1$) then determines the causal order between the following two parties ($A_2 \prec A_3$ or $A_3 \prec A_2$, respectively), where the second party must always output $0$ (corresponding to $a_2=0$ or $a_3=0$, resp.) and the last party must produce the input of the second party ($a_3=x_2$ or $a_2=x_3$, resp.) as output. This correlation can thus be established in a well-defined, although dynamical, causal order and is thus causal. One can check that it is indeed of the form~\eqref{def:causal_correlation} (with only one term in the sum, for $k=1$), but not of the form~\eqref{def:correlation_party_last}: there is indeed no party that always acts last (note that, since the correlation is deterministic, the sum in~\eqref{def:correlation_party_last} would also need to have only one term, which would single out a fixed last party).

Finally, a generalization of the previous example is a situation in which the order between $A_2$ and $A_3$ is chosen probabilistically with a probability depending on the input of $A_1$. An example of this type is
\ba
P&&(a_2,a_3|x_1,x_2,x_3) = \nonumber \\
&& \quad \delta_{x_1,0}\left(q_0\,\delta_{a_2,0}\ \delta_{a_3,x_2} + \left(1-q_0\right)\, \delta_{a_2,x_3} \ \delta_{a_3,0}\right) \nonumber \\
&& \quad +\, \delta_{x_1,1}\left(q_1\,\delta_{a_2,0}\ \delta_{a_3,x_2} + \left(1-q_1\right)\, \delta_{a_2,x_3} \ \delta_{a_3,0}\right) ,
\label{probabilisticdynamical}
\ea
with $0< q_0, q_1 < 1$. In this example, $A_2 \prec A_3$ with probability $q_0$ if $x_1=0$ and with probability $q_1$ if $x_1=1$. Once again, this correlation is of the form~\eqref{def:causal_correlation} and can be established in a well-defined causal order. However, for $q_0\not=q_1$, it is not a probabilistic mixture of fixed-order causal correlations.\footnote{
To see this, assume that we can write $P=q P_1 + (1-q)P_2$, $q\in(0,1)$, where $P_1$ and $P_2$ are causal correlations with the fixed orders $A_1\prec A_2\prec A_3$ and  $A_1 \prec A_3\prec A_2$, respectively.
(Since $A_1$ has no output, we can assume they act first; see Sec.~\ref{subsec:trivialinout}.)
Note that $A_2$'s marginal distribution satisfies $P(a_2=0|000)=1$, and thus $P_1(a_2=0|000)=1$ also.
But the causal order of $P_1$ requires $P_1(a_2|000)=P_1(a_2|001)$, so that $q_0=P(a_2=0|001)=q + (1-q)P_2(a_2=0|001)\ge q$.
Similarly, since $P(a_3=1|000)=0$ we have $P_2(a_3=1|000)=P_2(a_3=1|010)=0$, so that $q_0=P(a_3=1|010)=qP_1(a_3=1|010)\le q$.
Together, this implies $q=q_0$.
Analogous reasoning for when $x_1=1$ implies that $q=q_1$ and thus we must have $q_0=q_1$ if $P$ is a mixture of fixed-order causal correlations.
}

\section{Characterization of causal correlations as a convex polytope}

As noted earlier, any convex combination of causal correlations is causal, meaning that causal correlations (for a given scenario) form a convex set. 
It was already argued in Refs.~\cite{oreshkov15,branciard16}, more precisely, that this set is a convex polytope, the so-called `causal polytope'.
Here we will prove this more explicitly by showing how any causal correlation can be written as a convex combination of deterministic causal correlations. 
(This was already proved for the bipartite case in Ref.~\cite{branciard16}.)
The polytope structure then follows from the fact that, for any given scenario, the number of such deterministic causal correlations is finite.
The facets of the causal polytope can be expressed as linear inequalities that are satisfied by all causal correlations: when nontrivial, these correspond to (tight) `causal inequalities'~\cite{oreshkov12,branciard16}.

\subsection{Decomposing causal correlations into deterministic ones}

Let us first introduce some more notation. A correlation is deterministic if the list of outputs, $\vec a$, is a deterministic function $\vec \alpha$ of the list of inputs, $\vec x$: $\vec a = \vec \alpha(\vec x)$. We shall then denote the corresponding probability distribution by $P_{\vec \alpha}^\textup{det}$, such that
\ba
P_{\vec \alpha}^\textup{det}(\vec a|\vec x) = \delta_{\vec a, \vec \alpha(\vec x)} \,.
\ea

We will now prove the following theorem:
\begin{thm}
\label{main}
Any $N$-partite causal correlation can be written as a convex combination
\ba
P(\vec a|\vec x) = \sum_{\vec \alpha} \, q_{\vec \alpha} \ P_{\vec \alpha}^\textup{det}(\vec a|\vec x) \label{eq:causal_deterministic_combination}
\ea
with $q_{\vec \alpha} \ge 0$, $\sum_{\vec \alpha} \, q_{\vec \alpha} = 1$,
where the sum is over all functions $\vec \alpha: \vec x \mapsto\vec a$ that define a deterministic causal correlation $P_{\vec \alpha}^\textup{det}(\vec a|\vec x)$.
\end{thm}

\medskip

The proof is by induction:

\begin{itemize}

\item For $N=1$, it is a well-known fact that any correlation can be written as a convex combination of deterministic ones (see, e.g., Ref.~\cite{fine82}), and any single-party correlation is causal.

\item For any given $N \ge 2$ we shall prove the following implication: if it is true that all $(N{-}1)$-partite causal correlations can be written as convex combinations of deterministic ones (the induction hypothesis), then the same is true for $N$-partite causal correlations.

\end{itemize}

Consider an $N$-partite causal correlation $P(\vec a|\vec x)$, decomposed in the form~\eqref{def:causal_correlation}, with the correlations $P_{k,x_k,a_k}(\vec a_{\backslash k}|\vec x_{\backslash k})$ being $(N{-}1)$-partite causal correlations (for all $k,x_k,a_k$).
By the induction hypothesis, the latter can be decomposed as in Eq.~\eqref{eq:causal_deterministic_combination}:
\begin{align}
\hspace{-3mm}P_{k,x_k,a_k}(\vec a_{\backslash k}|\vec x_{\backslash k}) = \sum_{\vec \alpha_{\backslash k}} q_{\vec \alpha_{\backslash k}}(k,x_k,a_k)  P_{\vec \alpha_{\backslash k}}^\textup{det}(\vec a_{\backslash k}|\vec x_{\backslash k}),\! \label{eq:causal_deterministic_combination_Nminus1}
\end{align}
where the weights $q_{\vec \alpha_{\backslash k}}(k,x_k,a_k)$ depend in general on $k,x_k,a_k$, and the sum is over all functions $\vec \alpha_{\backslash k}: \vec x_{\backslash k} \mapsto \vec a_{\backslash k}$ that define a deterministic causal correlation $P_{\vec \alpha_{\backslash k}}^\textup{det}(\vec a_{\backslash k}|\vec x_{\backslash k})$. This decomposition does not yet prove the theorem, because we need to express $P(\vec a|\vec x)$ as a convex combination with weights that do not depend on the inputs and outputs. However, we can remove this dependency by appropriately rearranging the sum \eqref{eq:causal_deterministic_combination_Nminus1}. To this end, we shall first prove the following lemma:

\begin{lem} \label{lem1}
Consider a set of $M$ points $Q_m$ ($m = 1, \ldots, M$) belonging to some linear space, and $Z$ different points $P(z)$ ($z = 1, \ldots, Z$) in their convex hull, written as convex combinations of the extremal points $Q_m$ in the following way:
\ba
P(z) = \sum_{m=1}^M \, q_m(z) \ Q_m \,,
\ea
with weights $q_m(z)$ that depend on $z$ (such that, for each $z$, all $q_m(z) \ge 0$ and $\sum_{m=1}^M q_m(z) = 1$).

Then, each point $P(z)$ can also be written as 
\ba
P(z) = \sum_{m_1= 1}^M \cdots \sum_{m_Z= 1}^M \tilde{q}_{m_1,\dots, m_Z} \ Q_{m_z} \,,
\ea
where it is now the extremal points $Q_{m_z}$ that depend on $z$, while the new weights $\tilde{q}_{m_1,\dots, m_Z} \ge 0$, $\sum_{m_1,\dots, m_Z} \tilde{q}_{m_1,\dots, m_Z} = 1$ are fixed.
\end{lem}

\begin{proof}
The new weights are defined as 
\ba
\tilde{q}_{m_1,\dots, m_Z} := \prod_{z=1}^Z q_{m_{z}}(z) \,.
\ea 
Then for a given $z$,
\ba
\sum_{\tiny \begin{array}{r}
m_1, \ldots, m_{z-1}, \\
m_{z+1}, \ldots, m_Z
\end{array}} \!\!\!\tilde{q}_{m_1,\dots, m_Z} \ = \ q_{m_z}(z) \,,
\ea
and
\ba
\sum_{m_1, \ldots, m_Z} \!\!\! \tilde{q}_{m_1,\dots, m_Z} \ Q_{m_z} \ &=& \ \sum_{m_z} \sum_{\tiny \begin{array}{r}
m_1, \ldots, m_{z-1}, \\
m_{z+1}, \ldots, m_Z
\end{array}} \!\!\! \tilde{q}_{m_1,\dots, m_Z} \ Q_{m_z} \nonumber \\
&=& \ \sum_{m_z} q_{m_z}(z) \ Q_{m_z} = P(z) \,, \quad
\ea
as required.
\end{proof}

Returning to the proof of Theorem~\ref{main}, we rename the party-input-output variables as $(k,x_k,a_k)\equiv z_k=1,\dots,Z_k$. 
We can now apply Lemma~\ref{lem1} to Eq.~\eqref{eq:causal_deterministic_combination_Nminus1} and write
\begin{multline}
P_{k,x_k,a_k}(\vec a_{\backslash k}|\vec x_{\backslash k})  \\
=
\sum_{\vec \alpha_{\backslash k}^{1},\dots,\vec \alpha_{\backslash k}^{Z_k}} \,
\tilde{q}_{{\vec \alpha_{\backslash k}^{1},\dots,\vec \alpha_{\backslash k}^{Z_k}}} \ 
P_{\vec \alpha_{\backslash k}^{z_k}}^\textup{det}(\vec a_{\backslash k}|\vec x_{\backslash k}) \,, \label{eq:causal_deterministic_combination_Nminus1_v2}
\end{multline}
where the correlations $P_{\vec \alpha_{\backslash k}^{z_k}}^\textup{det}(\vec a_{\backslash k}|\vec x_{\backslash k})$ are taken from the same set as the $P_{\vec \alpha_{\backslash k}}^\textup{det}(\vec a_{\backslash k}|\vec x_{\backslash k})$'s above, and hence are deterministic and causal.

The single-party probability distributions $P_k(a_k|x_k)$ in Eq.~\eqref{def:causal_correlation} can also be decomposed as a combination of deterministic correlations,
\ba
P_k(a_k|x_k) = \sum_{\alpha_k} \, q'_{\alpha_k} \ P_{\alpha_k}^\textup{det}(a_k|x_k) \,. \label{eq:causal_deterministic_combination_N_k}
\ea

Using Eqs.~\eqref{eq:causal_deterministic_combination_Nminus1_v2} and~\eqref{eq:causal_deterministic_combination_N_k}, we can now expand correlations $P(\vec a|\vec x)$ of the form of~\eqref{def:causal_correlation} as
\begin{multline}
P(\vec a|\vec x) \, = \, \sum_{k \in {\cal N}} \ q_k \ \sum_{\alpha_k} \, q'_{\alpha_k} \ P_{\alpha_k}^\textup{det}(a_k|x_k)  \\ 
\qquad \times \hspace{-8pt}
 \sum_{\vec \alpha_{\backslash k}^{1},\dots,\vec \alpha_{\backslash k}^{Z_k}} \, 
\tilde{q}_{{\vec \alpha_{\backslash k}^{1},\dots,\vec \alpha_{\backslash k}^{Z_k}}} \ 
P_{\vec \alpha_{\backslash k}^{z_k}}^\textup{det}(\vec a_{\backslash k}|\vec x_{\backslash k}) \\
\hspace{-5pt} = \hspace{-15pt}
 \sum_{\tiny \begin{array}{c}
	k,\alpha_k, \\ 
	\vec \alpha_{\backslash k}^{1},\dots,\vec \alpha_{\backslash k}^{Z_k}
\end{array}} \hspace{-15pt}
q_k \, q'_{\alpha_k} \, \tilde{q}_{{\vec \alpha_{\backslash k}^{1},\dots,\vec \alpha_{\backslash k}^{Z_k}}}
 \, P_{\alpha_k}^\textup{det}(a_k|x_k) \, P_{\vec \alpha_{\backslash k}^{z_k}}^\textup{det}(\vec a_{\backslash k}|\vec x_{\backslash k}) \,. \!\!\!
\end{multline}
This is indeed a convex combination of deterministic causal correlations, with weights independent of inputs and outputs, which thus completes the proof.
\qed

\subsection{Describing deterministic causal strategies}
\label{subsec:strategies}

As mentioned above, a deterministic `strategy' (or correlation) can be characterized by a deterministic function $\vec \alpha$ of the list of inputs $\vec x$, which determines the list of outputs $\vec a = \vec \alpha(\vec x)$.

Of course, not any such function will make the correlation $P_{\vec \alpha}^\textup{det}(\vec a|\vec x) = \delta_{\vec a, \vec \alpha(\vec x)}$ causal. In order to be causal, $P_{\vec \alpha}^\textup{det}$ must indeed have a decomposition of the form~\eqref{def:causal_correlation}.
Since $P_{\vec \alpha}^\textup{det}(\vec a|\vec x)$ can only take values 0 or 1, this implies in particular that the weights $q_k$ are also 0 or 1 and hence, there is only one term in the sum.

That is, the causal deterministic strategy $\vec \alpha$ can be understood as follows: it determines a party $A_{k_1}$ that acts first. The output $a_{k_1}$ of that party is then a deterministic function of its input $x_{k_1}$ (which is also specified by $\vec \alpha$). For each input $x_{k_1}$ of that party (and the corresponding deterministic output $a_{k_1}$), the remaining parties must then also share a deterministic and causal correlation---which in turn must be compatible with one specific party acting first. Hence, the input $x_{k_1}$ of the first party also determines the party $A_{k_2(x_{k_1})}$ that acts second (recall that causal correlations allow for dynamical causal orders, see the example in Sec.~\ref{subsubsec:examples}); the response function of that party is then a deterministic function of the input of the first party and its own input. Continuing in this fashion, the party that acts third then depends on the inputs of $A_{k_1}$ and $A_{k_2(x_{k_1})}$, and its output is a deterministic function of the inputs of those two parties and its own input; etc.

Thus, each given set of inputs $\vec x$ can be viewed as being processed in a particular causal order\footnote{If some parties are causally independent then this order may not be unique.} $A_{k_1} \prec A_{k_2^{\vec x}} \prec \cdots \prec A_{k_N^{\vec x}}$ (with $k_2^{\vec x} = k_2(x_{k_1})$, etc.), so that the correlation $P_{\vec \alpha}^\textup{det}(\vec a|\vec x)$ can be written as
\ba
P_{\vec \alpha}^\textup{det}(\vec a|\vec x) &=& P(a_{k_1}|x_{k_1}) \, P(a_{k_2^{\vec x}}|x_{k_1},x_{k_2^{\vec x}})\notag \\ 
 &&  \times \, \cdots \, \times\, P(a_{k_{N-1}^{\vec x}}| \vec x_{\backslash k_N^{\vec x}}) \, P(a_{k_N^{\vec x}}|\vec x) \, . \quad
\ea
It follows in particular that for each given $\vec x$, there exist $N$ nested subsets ${\cal  K}_1 \subset {\cal  K}_2^{\vec x} \subset \cdots \subset {\cal  K}_N^{\vec x}$ with ${\cal  K}_K^{\vec x} = \{ k_1, k_2^{\vec x}, \ldots, k_K^{\vec x} \}$, $|{\cal  K}_K^{\vec x}| = K$, such that for all $K$, the marginal distribution
\ba
P_{\vec \alpha}^\textup{det}(\vec a_{{\cal  K}_K^{\vec x}}|\vec x) = P_{\vec \alpha}^\textup{det}(\vec a_{{\cal  K}_K^{\vec x}}|\vec x_{{\cal  K}_K^{\vec x}})
\ea
does not depend on the inputs of the last $N-K$ parties in the causal order realized on input $\vec x$.

\subsection{Causal correlations in scenarios with trivial inputs or outputs}\label{subsec:trivialinout}

Consider a scenario in which one party $A_k$ has a fixed output for all its possible inputs (the output being fixed, we could just ignore it and equivalently say that $A_k$ has no output).\footnote{This scenario also arises when one averages over $A_k$'s outputs.}
It is well known that, for Bell-type local correlations, this scenario is equivalent to the similar one in which $A_k$ is simply ignored.
More generally, if a single input $x_k$ has a fixed (or, equivalently, no) output, then the local polytope is simply equivalent to that obtained by discarding the input $x_k$ completely~\cite{pironio05}.
In contrast, for the case of causal inequalities it has already been noted that one can obtain interesting correlations with `nontrivial inputs with fixed outputs'~\cite{branciard16}.
What can we therefore say more generally about the causal polytope when $A_k$ has a fixed output for \emph{all} its inputs?

In such a scenario we can write the $N$-partite correlation as
\begin{equation}
	P(\vec{a}|\vec{x}) = P(\vec{a}_{\backslash k}|\vec{x}_{\backslash k}, x_k) = P_{x_k}(\vec{a}_{\backslash k}|\vec{x}_{\backslash k})\,.
\end{equation}
If $P_{x_k}(\vec a_{\backslash k}|\vec x_{\backslash k})$ is causal for all $x_k$ then $P(\vec a|\vec x)$ is trivially of the form~\eqref{def:causal_correlation}, and therefore causal.
Conversely, if $P(\vec a|\vec x)$ is causal then, by the remark discussed in Subsection~\ref{subsubsec:ignore_parties}, $P_{x_k}(\vec a_{\backslash k}|\vec x_{\backslash k})$ is also causal for each $x_k$.

Thus, the $N$-partite correlation is causal if and only if all of the conditional $(N{-}1)$-partite correlations obtained for each possible input $x_k$ of $A_k$ are causal.
In order to test whether $P(\vec{a}|\vec{x})$ is causal it therefore suffices to test whether the $(N{-}1)$-partite correlations are causal, and one can always assume that $A_k$ is located before all the other parties.\footnote{Hence, in an $N$-partite scenario where one party has a trivial output, a noncausal correlation can only be obtained if some reduced $(N{-}1)$-partite correlation is already noncausal. Note that in contrast, in the framework of process matrices~\cite{oreshkov12} (see Section~\ref{sec:process_matrices}), the property of causal nonseparability of a process---which is the `device-dependent' analog of noncausality for correlations~\cite{araujo15,branciard16}---can be witnessed in a scenario where some parties only have trivial outputs (e.g., where they simply implement unitary operations), while all reduced processes involving fewer parties are causally separable~\cite{araujo15,branciard16b}.}

\medskip

Another important scenario to understand is that in which a party $A_k$ has a single fixed input (or equivalently, as before, no input).
In this case we have (from the definition of conditional probabilities)
\begin{align}
\hspace{-1mm} P(\vec a|\vec x) = P(\vec a_{\backslash k}, a_k|\vec x_{\backslash k}) = P(\vec a_{\backslash k}|\vec x_{\backslash k}) P_{\vec x_{\backslash k},\vec a_{\backslash k}}\!(a_k),
\end{align}
with $P_{\vec x_{\backslash k},\vec a_{\backslash k}}(a_k) := P(a_k|\vec x_{\backslash k},\vec a_{\backslash k})$.
If $P(\vec a_{\backslash k}|\vec x_{\backslash k})$ is causal then $P(\vec a|\vec x)$ is clearly of the form~\eqref{def:correlation_party_last}, and thus causal.
Conversely, referring again to the remark in Sec.~\ref{subsubsec:ignore_parties}, if $P(\vec a|\vec x)$ is causal then so is $P(\vec a_{\backslash k}|\vec x_{\backslash k})$.

Thus, as is the case for locality, the causality of the $N$-partite correlation is equivalent to the causality of the $(N{-}1)$-partite correlation obtained by discarding the party $A_k$ with a fixed input.
Causally, one may consider that $A_k$ always acts after the other $N-1$ parties.
Note that this is also true, more generally, whenever a party $A_k$ cannot signal to any other party or set of parties---as is indeed the case when they have a fixed (or no) input.

\section{Simplest tripartite inequalities}\label{sec:tripartiteinequalities}

With the basic properties of the multipartite causal polytope laid out, we wish to study in detail the simplest scenario with more than two parties (i.e., which is not reducible to the bipartite scenario that was characterized in Ref.~\cite{branciard16}).
In contrast to the case for Bell inequalities, where the simplest such case is the `$(3,2,2)$' scenario with 3 parties all having binary inputs and outputs, the discussion in the previous section suggests that a simpler tripartite scenario exists for causal correlations.
This is the scenario where each party $A_k$ has a binary input $x_k$, a single constant output for one of the inputs, and a binary output for the other.
Specifically, we consider that for each $k$, the input $x_k=0$ has the constant output $a_k=0$, while the input $x_k=1$ has two possible outputs, $a_k=0$ or~$1$.

As is standard, let us denote the three parties $A, B, C$ (i.e., $A_1=A, A_2=B, A_3=C$), their inputs $x,y,z$ (instead of $x_1,x_2,x_3$) and their outputs $a,b,c$ (instead of $a_1,a_2,a_3$).
We will denote below by $P_{ABC}$ the complete tripartite probability distribution (i.e., $P_{ABC}(a,b,c|x,y,z) := P(a,b,c|x,y,z)$), and by $P_{AB}$, $P_A$, etc. the marginal distributions for the parties indicated by the subscript (e.g., $P_{AB}(a,b|x,y,z) := \sum_c P_{ABC}(a,b,c|x,y,z)$, etc.). Note that every marginal distribution retains a dependency on all three inputs.

\subsection{Characterizing the causal polytope}

The vertices of the causal polytope for this scenario can be found by enumerating all the deterministic probability distributions $P_{ABC}(a,b,c|x,y,z)$ compatible with any of the 12 possible definite causal orders (for each of the 3 parties acting first, there are 4 possible causal orders for the remaining 2 parties: two fixed orders, and two dynamical ones, where the order depends on the input of the first party).  One finds that there are 680 such strategies (and thus vertices), of which 488 are compatible with a fixed causal order, while the remaining 192 require a dynamical order to be realized.

The causal polytope is $19$-dimensional, since this is the minimum number of parameters needed to completely specify any probability $P_{ABC}(a,b,c|x,y,z)$: for each set of inputs $x,y,z$, if $n$ of them are non-zero then one needs $2^n-1$ values to specify the probabilities completely for these inputs (normalization determines the remaining value)---so that the dimension of the problem is indeed $\sum_{n=0}^{3}\binom{3}{n}(2^n-1)=19$.

In order to determine the facets of this polytope, which correspond directly to tight causal inequalities, a parametrization of the polytope must be fixed and the convex hull problem solved~\cite{branciard16}.
Several such parametrizations of $P_{ABC}(a,b,c|x,y,z)$ are possible but we found that, because of the size of the polytope, the ability to solve the convex hull problem depended critically on the chosen parametrization.
Using the software \textsc{cdd}~\cite{cdd} we were able to compute the facets of the polytope from its description in terms of its vertices with the following parametrization:
\begin{align}
\vec P = \big( & P_A(1|100), P_B(1|010), P_C(1|001), \notag \\
& P_{AB}(10|110), P_{AB}(01|110), P_{AB}(11|110), \notag \\
& P_{BC}(10|011), P_{BC}(01|011), P_{BC}(11|011), \notag \\
& P_{AC}(01|101), P_{AC}(10|101), P_{AC}(11|101), \notag \\
& P_{ABC}(100|111), P_{ABC}(010|111), \notag \\
& \quad P_{ABC}(001|111), P_{ABC}(110|111), \notag \\
& \qquad P_{ABC}(011|111), P_{ABC}(101|111), \notag \\
& \qquad \qquad \qquad \qquad \qquad P_{ABC}(111|111) \big) \,.
\end{align}

In total, the polytope was found to have $13\,074$ facets, each corresponding to a causal inequality.
However, inequalities that can be obtained from one another, either by relabeling outputs or permuting parties, can be considered to be equivalent.
Once such equivalences are taken into account one finds that there are 305 equivalence classes, or `families', of inequalities.
A complete list of these families can be found in the Supplemental Material~\cite{SM}, but in what follows we will focus on some specific interesting examples.

\subsection{Three simple inequalities}\label{subsec:threeineq}

As is standard for polytopes of correlations, several facets correspond to trivial inequalities of the form $P(a,b,c|x,y,z) \ge 0$.
Specifically, there are three inequivalent such families, corresponding to 1, 2 or 3 inputs being 1.
One also recovers conditional versions of the nontrivial bipartite `lazy guess-your-neighbor's-input' (LGYNI) inequalities $P(x ( a \oplus y ) = y (b \oplus x) = 0 |z) \le 3/4$, where $\oplus$ denotes addition modulo 2 (and where this notation implicitly assumes that the inputs $x$ and $y$ are uniformly distributed)~\cite{branciard16}.
These can equivalently be written as
\begin{equation}\label{ineq:lgyni}
P_A(1|10z) + P_B(1|01z) - P_{AB}(11|11z) \ \ge \ 0 \,. 
\end{equation}
There are two inequivalent families of such inequalities, for the two cases of $z = 0$ or $1$.

Amongst the remaining families of inequalities there are 5 which are completely symmetric under exchange of parties, and which are good candidates for simple inequalities that nontrivially involve all three parties.
The following two are of particular interest due to their simple form, and the fact that they can be seen as natural generalizations of the LGYNI inequality~\eqref{ineq:lgyni}:
\begin{align}
I_1 =&\, P_{AB}(11|110) + P_{BC}(11|011) \notag \\ &+ P_{AC}(11|101)
- P_{ABC}(111|111) \ \ge \ 0 \,, \label{ineq:1}
\end{align}
and
\begin{align}
I_2 =&\, P_A(1|100) + P_B(1|010) \notag \\
&+ P_C(1|001) - P_{ABC}(111|111)  \ \ge \ 0 \,. \label{ineq:2}
\end{align}

As is the case for the LYGNI inequality~\eqref{ineq:lgyni}, these two inequalities can be expressed as `causal games'~\cite{oreshkov12,branciard16}.
They can indeed be written as 
\begin{align}
& P\big(xy(ab \oplus z) = yz(bc \oplus x) \notag \\
&\hspace{18mm} =xz(ac \oplus y)=0\big) \ \le \ 7/8 \label{ineq:1_game}
\end{align}
and
\begin{align}
& P\big(x(y \oplus z \oplus 1)(a \oplus yz) =y(x \oplus z \oplus 1)(b \oplus xz) \notag \\
& \hspace{18mm} =z(x \oplus y \oplus 1)(c \oplus xy)=0\big) \ \le \ 7/8 \,, \label{ineq:2_game}
\end{align}
respectively, where it is implicitly assumed that all inputs occur with the same probability.
More precisely, the first inequality can be interpreted as a game in which the goal is to collaborate so that, whenever two parties both receive the input 1, the product of their outputs should match the input of the other party (in all other cases any output wins the game).
The second inequality can be interpreted as a similar game in which the goal is to ensure that, whenever a party receives the input 1 and the other two parties receive the same input, that party's output should match the other parties' inputs.
In both cases, the probability of success can be no greater than $7/8$ if the three parties follow a causal strategy.
It is simple to saturate this bound with a deterministic causal strategy: for example, if all parties always output $0$, both games are won in all cases except that where all inputs are $1$, giving indeed a success probability of $7/8$.

Another, simple inequality of interest that is symmetric only under a cyclic permutation of parties, and can also be seen as a generalization of the LGYNI input inequality, is the following:
\begin{align}
 I_3 =\,& 2 - P_{AB}(01|110) - P_{BC}(01|011) \notag \\
&\hspace{30mm} - P_{AC}(10|101) \ \ge \  0 \,. \label{ineq:3}
\end{align}
As for the previous two inequalities, this causal inequality can be interpreted as a causal game in the form (still implicitly assuming a uniform distribution of inputs for all parties):
\begin{align}
& P\big(xy(z \!\oplus\! 1)((a \!\oplus\! 1) b \!\oplus\! 1)=yz(x \!\oplus\! 1)((b \!\oplus\! 1)c \!\oplus\! 1) \notag \\
& \hspace{17mm} =xz(y \!\oplus\! 1)((c \!\oplus\! 1) a \!\oplus\! 1)=0\big) \ \le \ 7/8 \,, \label{ineq:3_game}
\end{align}
where the goal of the game is to ensure that whenever exactly two parties receive the input 1, 
each of them must guess the input of their left-hand neighbor (where $C$ is considered, in a circular manner, to be to the left of $A$). 
The bound $7/8$ on the probability of success can for instance be reached with the causal strategy, compatible with the order $A\prec B\prec C$, where the parties output $a=0$, $b=xy$ and $c=yz$: this strategy indeed wins the game in all cases except when the inputs are $(x,y,z) = (1,0,1)$.

\subsection{Generalizing tripartite causal inequalities}

For scenarios more complicated than the `simplest' tripartite one considered above, the convex hull problem---and thus the characterization of the causal polytope---very quickly becomes intractable.
For example, the polytope for the `complete binary' tripartite case where binary outputs are allowed for both inputs, has $138\,304$ vertices and is 56-dimensional.
Beyond this, it moreover becomes difficult to even enumerate the different vertices of the causal polytope.

Although we were hence unable to enumerate the causal inequalities for this complete binary tripartite scenario, by enumerating the vertices of the polytope we were able to verify that the three causal inequalities $I_1, I_2, I_3 \ge 0$ discussed above are in fact facets in this scenario as well.
To see this, one can enumerate all the deterministic strategies and thus vertices of the polytope, and use the fact that an inequality is a facet of the polytope if and only if \emph{a}) it is satisfied by every vertex of the polytope, and \emph{b}) there are $d$ affinely independent vertices saturating the inequality, where $d$ is the dimension of the polytope~\cite{pironio05}.

Another facet of this complete binary tripartite polytope has independently been found by Ara\'{u}jo and Feix~\cite{AraujoFeixIneq} (reproduced here with their permission) and can be written as
\begin{align}
	I_4=\,& 2 - P_{ABC}(000|000)  - P_{ABC}(011|110) \notag \\ 
	& \quad  - P_{ABC}(101|011) - P_{ABC}(110|101) \ \ge \ 0 \,. \label{ineq:4}
\end{align}
This inequality can be interpreted as a causal game that generalizes that of Eq.~\eqref{ineq:3_game}:
\begin{align}
P\big( (x\!\oplus\! y \!\oplus\! z \!\oplus\! 1)(&(b\!\oplus\! x\!\oplus\! 1)(c \!\oplus\! y\!\oplus\! 1) \notag \\
&  
\times(a \!\oplus\! z\!\oplus\! 1)\!\oplus\! 1)=0\big) \le \ 3/4\,,\label{ineq:4_game}
\end{align}
where the goal of the game is, whenever an even number of parties receive the input 1, for every party to guess the input of their left-hand neighbor.\footnote{The bound $3/4$ on the probability of success can for instance be reached with the causal strategy, compatible with the order $A\prec B\prec C$, where the parties output $a=0$, $b=x$ and $c=y$: this strategy indeed wins the game in all cases except when the inputs are $(x,y,z) = (1,0,1)$ or $(0,1,1)$.}
This game is equivalent to a form of the original multipartite, cyclic, guess your neighbor's input (GYNI) game~\cite{Almeida:2010aa} in which each party must always guess their left-hand neighbor's input, but a non-uniform distribution of inputs is considered (namely, only the four input combinations appearing in Eq.~\eqref{ineq:4}  are allowed). It is interesting to note that the inequality corresponding to the variant of the GYNI game with a uniform distribution of \emph{all} inputs is not a facet of the polytope.

Thus, this game is a form of the original multipartite, cyclic, guess-your-neighbor’s-input (GYNI) game introduced in Ref.~\cite{Almeida:2010aa}, where different input distributions where considered.
It is interesting to note that the inequality corresponding to the GYNI game with a uniform distribution over all inputs is not a facet of the polytope.

As we noted earlier, the inequalities discussed in the previous subsection can be seen as natural possible generalizations of the LGYNI bipartite inequalities from Ref.~\cite{branciard16}.
This suggests that similar generalizations to $N$-partite scenarios might provide tight causal inequalities for arbitrarily many parties.
In particular, the natural generalizations of the first two inequalities, Eqs.~\eqref{ineq:1} and~\eqref{ineq:2}, to $N$-parties would be:\footnote{To prove that Eq.~\eqref{ineq:1gen} indeed defines a valid causal inequality for all $N$, it is sufficient to prove that it holds for all deterministic causal correlations. From the remark at the end of Subsection~\ref{subsec:strategies} we know that, given a deterministic causal correlation $P$, the input $\vec x = (1,\ldots,1)$ fixes a particular causal order, and in particular a `last' party $A_k$, such that $P_{\mathcal{N} \backslash k}(1,\dots, 1|\vec x_{\backslash k}=(1,\dots, 1), x_k=0) = P_{\mathcal{N} \backslash k}(1,\dots, 1|\vec x_{\backslash k}=(1,\dots, 1), x_k=1)$. This implies that $P_{\mathcal{N} \backslash k}(1,\dots, 1|\vec x_{\backslash k}=(1,\dots, 1), x_k=0) - P_\mathcal{N}(1,\dots,1 \,|\, 1,\dots, 1) \ge 0$. $J_{1}(N)$ is then obtained by adding some more nonnegative terms $P_{\mathcal{N} \backslash k'}(\ldots)$, so that it necessarily remains nonnegative.
\\
Similarly, to prove Eq.~\eqref{ineq:2gen}, we note that a deterministic causal correlation fixes a `first' party $A_k$ such that $P_k(1| x_k=1, \vec x_{\backslash k}=(0,\dots, 0)) = P_k(1| x_k=1, \vec x_{\backslash k}=(1,\dots, 1))$ and thus $P_k(1| x_k=1, \vec x_{\backslash k}=(0,\dots, 0)) - P_\mathcal{N}(1,\dots,1 \,|\, 1,\dots, 1) \ge 0$. Again, $J_{2}(N)$ is then obtained by adding some more nonnegative terms, so that it remains nonnegative.}
\begin{align}
J_{1}(N) =& \sum_{k\in{\cal N}} P_{\mathcal{N} \backslash k}\big(1,\dots, 1\,|\, x_k=0,\, \vec x_{\backslash k}=(1,\dots, 1)\big)\notag \\[-1mm]
& \hspace{10mm} - P_\mathcal{N}\big(1,\dots,1 \,|\, 1,\dots, 1\big)  \ \ge \ 0 \,, \label{ineq:1gen}
\end{align}
and
\begin{align}
J_{2}(N) =& \sum_{k\in{\cal N}} P_k \big(1\,|\, x_k=1,\, \vec x_{\backslash k}=(0,\dots, 0)\big)\notag \\[-1mm]
& \hspace{10mm} - P_\mathcal{N}\big(1,\dots, 1 \,|\, 1,\dots, 1\big) \ \ge \ 0 \,, \label{ineq:2gen}
\end{align}
which are both obtained by generalizing the game interpretation of the corresponding tripartite inequalities, replacing the references to `two parties' with `$N-1$ parties' in the descriptions of these games.
One finds that, with a causal strategy, one can win these games with probability\footnote{This probability of success (or equivalently, the bounds $0$ in Eqs.~\eqref{ineq:1gen} and~\eqref{ineq:2gen}) can again be reached when all parties always output $0$: they then lose the game only when all inputs are~$1$.} $1 - 1/2^{N}$, which approaches 1 as $N$ becomes large.
By enumerating all the vertices of the `simplest' 4-partite causal polytope (where again, on input $x_i=0$ the output is fixed as $a_i=0$, and on input $x_i=1$ binary output is allowed), which has $3\,209\,712$ vertices and is 65-dimensional, we were able to verify that both the inequalities $J_1(4)\ge 0$ and $J_2(4)\ge 0$ are facets of the corresponding causal polytope. It remains an open question to prove their tightness for $N\ge 5$.

\section{Violating causal inequalities with process matrix correlations}
\label{sec:process_matrices}

The facet inequalities presented in Section~\ref{sec:tripartiteinequalities} give bounds on the correlations that can be obtained causally, but is it possible to violate them in a more general framework in which causality is only assumed to hold locally?
In this section we use the process matrix formalism introduced by Oreshkov, Costa and Brukner~\cite{oreshkov12} to study precisely this possibility, thus performing a similar analysis to that of Ref.~\cite{branciard16} in showing the violation of causal inequalities with bipartite process matrices.

\subsection{The process matrix framework}

The process matrix framework allows one to formalize scenarios in which quantum mechanics holds locally for each party, but in which no global causal order between the parties is assumed.
Instead, weaker consistency conditions are imposed to ensure that probabilities are well behaved and paradoxes thus avoided.
We will recall only the basics of this formalism needed for what follows, and we refer the reader to Refs.~\cite{oreshkov12,araujo15} for a more detailed presentation of the framework.

Each party $A_k$ receives an incoming physical system described in a Hilbert space $\Hi^{A_k^I}$ and produces an outgoing system in a Hilbert space $\Hi^{A_k^O}$.
The operation that $A_k$ performs (which in general depends on the classical input $x_k$) is modeled by a quantum instrument~\cite{Davies:1970aa}.
An instrument is a set of completely positive~(CP) trace non-increasing maps from $\LO(\Hi^{A_k^I})$ to $\LO(\Hi^{A_k^O})$, where $\LO(\Hi)$ denotes the space of linear operators over $\Hi$, such that the sum of the maps is trace preserving.
Each CP map in an instrument is associated with a measurement outcome $a_k$.
Any such instrument can be conveniently represented as a set of operators $\{M_{a_k|x_k}^{A_k^IA_k^O}\}_{a_k}$ (for a fixed input $x_k$) on the space $\LO(\Hi^{A_k^I}\otimes\Hi^{A_k^O})$ using the Choi-Jamio\l{}kowski (CJ) isomorphism~\cite{Choi:1975aa,Jamiolkowski:1972aa}, satisfying the following criteria~\cite{oreshkov12}:
\begin{equation}\label{eqn:instrumentCond}
	\forall a_k ,\, M_{a_k|x_k}^{A_k^IA_k^O}\ge 0 \quad \text{and} \quad \tr_{A_k^O}\sum_{a_k}M_{a_k|x_k}^{A_k^I A_k^O}=\id^{A_k^I},
\end{equation}
where $\tr_{A_k^O}$ denotes the partial trace over $\Hi^{A_k^O}$ and $\id^{A_k^I}$ is the identity operator on $\Hi^{A_k^I}$.

By requiring the consistency and normalization of probabilities for all quantum instruments, it can be shown that the correlations obtainable in such a scenario can be calculated through the \emph{generalized Born rule}, which for the tripartite case can be written
\begin{align}
	& P(a,b,c|x,y,z)\notag \\
	& \qquad =\tr\left[ \left( M_{a|x}^{A^I A^O}\otimes M_{b|y}^{B^I B^O} \otimes M_{c|z}^{C^I C^O} \right)\cdot W \right],\label{eqn:genBornRule}
\end{align}
for some positive semidefinite Hermitian matrix $W\in\LO(\Hi^{A^I}\otimes\Hi^{A^O}\otimes\Hi^{B^I}\otimes\Hi^{B^O}\otimes\Hi^{C^I}\otimes\Hi^{C^O})$ called the \emph{process matrix}.
In order for such a $W$ to be a valid process matrix it must satisfy some further linear constraints.
We refer the reader to Refs.~\cite{araujo15,oreshkov15} for explicit presentation of these for tripartite process matrices as the conditions themselves are not particularly important for what follows.

Process matrices generalize the notions of quantum states and channels, and represent the interactions between the parties without enforcing any causal order between them.
It remains an open question whether all process matrices can be realized physically, but it is known that some processes which are incompatible with any definite causal order (e.g., the quantum switch mentioned earlier) can indeed be physically implemented~\cite{Chiribella:2013aa,araujo14,Procopio:2015aa}.
We refer to the correlations that can be realized by process matrices as \emph{process matrix correlations}.

\subsection{Violating the simplest tripartite inequalities}

In order to look for process matrices that can violate the tripartite causal inequalities we presented in Section~\ref{sec:tripartiteinequalities}, we generalize the `see-saw' approach successfully used in Ref.~\cite{branciard16} to find violations of bipartite causal inequalities.

For a causal inequality of the form $I\big(P(\vec{a}|\vec{x})\big)\ge 0$, to calculate the value of $I\big(P(\vec{a}|\vec{x})\big)$ obtained by a process matrix correlation one needs both the process matrix $W$ and the $N$ parties' instruments $\{M_{a_k|x_k}^{A_k^I A_k^O}\}$.
The algorithmic approach we used to find violations of such inequalities makes use of the fact that, given either the instruments for all $N$ parties or the instruments for $N-1$ parties and a valid process matrix, the problem of finding, respectively, the $W$ or the $N$th instrument that minimizes $I\big(P(\vec{a}|\vec{x})\big)$ can be expressed as a semidefinite programming (SDP) problem~\cite{branciard16} that can be solved efficiently.
The algorithm initially selects random instruments for all parties, then iteratively solves the SDP problem to find the optimal process matrix and instruments for each party in turn.
This iterative procedure continues until the algorithm converges to a value of $I\big(P(\vec{a}|\vec{x})\big)$.
Although this is only guaranteed to find a local, not global, minimum, by repeating the procedure many times with different initial instruments one can obtain a bound on the optimal violation (if any) of the causal inequality of interest.

Using this approach, we found processes matrices that can be used to violate all three causal inequalities~\eqref{ineq:1},~\eqref{ineq:2} and~\eqref{ineq:3} (i.e., $I_1, I_2, I_3 \ge 0$) presented in Subsection~\ref{subsec:threeineq} with two-dimensional incoming and outgoing systems for each party, i.e., with `qubits' (all $\Hi_k^{A^{I/O}} = \mathbb{C}^2$).

In order to violate the first inequality~\eqref{ineq:1}, one can take the process matrix $W_1\in\LO(\Hi^{A^I}\otimes\Hi^{A^O}\otimes\Hi^{B^I}\otimes\Hi^{B^O}\otimes\Hi^{C^I}\otimes\Hi^{C^O}) = \LO((\mathbb{C}^2)^{\otimes 6})$ to be
\ba
W_1 &=& \frac18 \Big[ \id^{\otimes 6} - \frac12 \big[ Z \id \id \id \id \id + \id \id Z \id \id \id + \id \id \id \id Z \id \nonumber \\
&& \hspace{18mm} - Z \id Z \id Z \id + (\id{-}Z) Z (\id{-}Z) Z Z \id \nonumber \\
&& \hspace{4mm} + Z \id (\id{-}Z) Z (\id{-}Z) Z + (\id{-}Z) Z Z \id (\id{-}Z) Z \big] \Big] , \nonumber \\ \label{eq:W_I1}
\ea
where $Z$ denotes the Pauli $Z$ matrix, $\id$ the $2 \times 2$ identity matrix, and where tensor products are implicit between all terms.
It can readily be verified that $W_1$ is indeed a valid process matrix~\cite{araujo15,oreshkov15}.
One can then take the instruments for all three parties to be:
\begin{gather}
\big\{ M_{0|0} = \frac12(\id + ZZ) = \proj{0} \otimes \proj{0} + \proj{1} \otimes \proj{1} \big\} \,, \notag \\
\big\{M_{0|1} = \proj{0} \otimes \proj{1} \,, \ M_{1|1} = \proj{1} \otimes \proj{0}\big\} \,, \label{eq:instruments1}
\end{gather}
where $\{\ket{0},\ket{1}\}$ is the computational basis (i.e., the eigenbasis of the $Z$ operator).
These instruments are easily seen to satisfy Eq.~\eqref{eqn:instrumentCond} and have simple interpretations: on input $0$ each party performs a computational basis measurement, prepares and sends the eigenstate corresponding to the measurement outcome, and outputs $0$ (as they must for this input); on input $1$, a computational basis measurement is again performed, the party outputs the result ($0$ or $1$) obtained from that measurement, and sends the basis state corresponding to the \emph{other} (unobtained) measurement outcome.
With this process matrix and these instruments, one finds, from the generalized Born rule~\eqref{eqn:genBornRule}, that 
\begin{equation}
	I_1=-1<0,
\end{equation}
which is the maximum algebraic violation of~\eqref{ineq:1}.
In other words, it wins the corresponding game (Eq.~\eqref{ineq:1_game}) perfectly (with probability $1$). 
Such maximal violation of a causal inequality was also found for a tripartite scenario in Ref.~\cite{baumeler13}, although this scenario differed in that all parties shared an additional classical input, and all outputs were binary.

It is interesting to note that both $W_1$ and the instruments $\{M_{a_k|x_k}\}$ are all diagonal in the computational basis. These can thus be interpreted as a classical process and classical operations, respectively \cite{Baumelerspace2016}.
If all parties perform classical operations it is known that they cannot violate any bipartite causal inequality~\cite{oreshkov12}.
For three-or-more parties, on the other hand, it has recently been shown that one can violate causal inequalities with purely classical operations and processes~\cite{baumeler14}, and our classical violation of Eq.~\eqref{ineq:1} is a novel illustration of this, in a simpler scenario than considered previously.

The third inequality~\eqref{ineq:3}, that we presented in Subsection~\ref{subsec:threeineq} can also be violated with exactly the same instruments as the first, using the process matrix
\begin{align}
W_3 =& \frac18 \Big[ \id^{\otimes 6} + \frac12 \big[ \id (\id{+}Z)Z \id Z Z - \id (\id{-}Z)Z Z Z \id \notag \\
& \qquad \qquad + Z Z \id (\id{+}Z) Z \id - Z \id \id (\id{-}Z)Z Z \notag \\
& \qquad \qquad + Z \id Z Z \id (\id{+}Z) - Z Z Z \id \id (\id{-}Z) \big] \Big], \quad \label{eq:W_I3}
\end{align}
which can again be readily verified to be valid.
As for the first inequality, one obtains the maximal algebraic violation $I_3=-1<0$ with a classical process matrix and classical instruments.

The second inequality~\eqref{ineq:2}, unlike the two discussed above, does not seem to be violated either maximally or by a classical process.
The best violation we found (with qubits) using the iterated optimization algorithm is obtained using the instruments:
\begin{gather}
\big\{ M'_{0|0} =\KetBra{\id}{\id} \big\} \,, \notag \\
\big\{ M_{0|1} = \proj{0} \otimes \proj{1} \,, \ M_{1|1} = \proj{1} \otimes \proj{0} \big\} \,, \label{eq:instruments2}
\end{gather}
where $\Ket{\id}=\ket{00}+\ket{11}$ and $\KetBra{\id}{\id}$ is the CJ representation of the identity channel.
The operation performed on input $1$ is thus the same as that previously considered, while on input $0$ each party outputs $0$ (as they must) and just sends out the physical system they received, unaffected. (Note that when acting on classical systems only, this is equivalent to the effect of the previous instrument $\{M_{0|0}\}$, so that $\{M'_{0|0}\}$ could also have been used to violate the previous two inequalities.)
With these instruments we found a process matrix $W_2$ giving
\begin{equation}
	I_2 \simeq -0.3367 < 0,
\end{equation}
which indeed violates Eq.~\eqref{ineq:2}.
$W_2$ cannot be expressed in the $Z$ basis alone and does not have a nice form so we do not write it explicitly here; it is included in full in the Supplemental Material~\cite{SM}.

It remains an open question to find the optimal violation of this inequality, and in particular whether it is possible to obtain a greater violation using higher dimensional systems.
For the violation of bipartite inequalities reported in Ref.~\cite{branciard16}, higher dimensional systems lead to greater violation for some inequalities, but the increased size of the tripartite optimization problem prevented us from investigating this possibility.

While the three tripartite inequalities we have focused on in this paper are of particular interest because of their simple form and their interpretation as natural generalizations of the LGYNI bipartite inequalities, a further search in fact showed that all $302$ nontrivial families of causal inequalities for the simplest tripartite scenario can be violated using the \emph{same} instruments given in Eq.~\eqref{eq:instruments2}.
Moreover, we found violations of all except $18$ families by completely classical process matrices (although often by less than with non-classical ones), and that the maximum algebraic violation can be obtained for $65$ families of inequalities (in all cases with classical process matrices and the instruments of Eq.~\eqref{eq:instruments1}). 
The processes violating the other inequalities, as well as the maximum violations we found, are included in the Supplemental Material~\cite{SM}.
We note that the instruments of Eq.~\eqref{eq:instruments2} are similar, but not equivalent to those used in the bipartite scenario in Ref.~\cite{branciard16}; those instruments can be used to violate some, but not all, of the tripartite causal inequalities, and often give smaller violations of the inequalities than the instruments of Eq.~\eqref{eq:instruments2}.
Finally, a process matrix (with a relatively simple form) and instruments violating inequality~\eqref{ineq:4} are also included in the Supplemental Material~\cite{SM}.

\subsection{Genuinely multipartite noncausal correlations}

An important question that arises with multipartite noncausal correlations---just as it does for nonlocal correlations---is whether the correlations are \emph{`genuinely $N$-partite noncausal'}, in the sense that no subset of the parties can have a definite causal relation to any other subset.
For the tripartite inequalities discussed above, one can test whether they can be violated by a process matrix in which one party (without loss of generality, $C$) causally follows (or precedes) the other two, but no definite causal order exists between $A$ and $B$, a scenario which we may denote $AB\prec C$ (or $C\prec AB$).
For the case $AB\prec C$ for instance, this formally corresponds to requiring that $C$ cannot signal to $A$ and $B$, which translates into the condition\footnote{For the case $C\prec AB$, the condition that $A$ and $B$ cannot signal to $C$ can also be expressed as a linear constraint on $W$, although it is a bit more complicated than~\eqref{eq:AB_prec_C}.}
\begin{equation}
	W = \frac{\id^{C^O}}{d_{C^O}}\otimes \tr_{C^O}W \,, \label{eq:AB_prec_C}
\end{equation}
where $d_{C^O}$ is the dimension of $\Hi^{C^O}$. This can easily be added as a linear constraint to the iterative SDP optimization algorithm, and by doing so we were able to find process matrices satisfying this constraint that violate many, but not all, of the tripartite causal inequalities.
For example, while no violation of Eq.~\eqref{ineq:1} was found, both Eqs.~\eqref{ineq:2} and~\eqref{ineq:3} can be violated (for qubit systems) with $I_2 = I_3 \simeq -0.2776 < 0$ by process matrices given in the Supplemental Material~\cite{SM} and the instruments in Eq.~\eqref{eq:instruments2}.
This violation is smaller than what we found for the unconstrained problem, but nonetheless shows that the inequalities $I_2, I_3 \ge 0$ do not detect genuinely tripartite noncausal correlations: a strictly lower bound than $0$ would need to be violated to detect those.

The study of genuinely multipartite nonlocal correlations has shown that there are several subtle issues associated with defining genuinely multipartite notions~\cite{Bancal:2013aa,Gallego:2012aa}.
Thus, while this situation is of great interest, we leave the problem of defining and characterizing genuinely multipartite noncausal correlations to future research.

\section{Conclusion}

The multipartite correlations that can be generated by localized operations within a definite causal order have a rich and complex structure. The fact that a party can, as a function of their input, influence the probabilities for the causal relations between parties in their future makes, in principle, the study of such correlations rather cumbersome. Here we have drastically simplified this problem, showing that all possible strategies can be reduced to deterministic ones---where both the outputs of and the causal relations between parties are functionally determined by the inputs of parties in their past---and probabilistic mixtures thereof.

Using this characterization, the set of causal correlations can be conveniently described as a convex polytope whose vertices represent deterministic strategies. This polytope can equivalently be characterized in terms of its facets, which define inequalities that have to be satisfied by the correlations. Taken together, all the inequalities (for a given scenario defining the parties and their sets of possible inputs and outputs) provide necessary and sufficient conditions for correlations to be realizable within a definite causal structure. We studied in detail the simplest tripartite scenario and found the corresponding complete set of facet inequalities. As this set is rather large---$13\,074$ inequalities, grouped into $305$ equivalence classes---we discussed explicitly only three examples. The full set may nonetheless be useful for systematic searches of violation of causal order. 

Finally, we found explicit violations of some exemplary inequalities within the process matrix formalism. Our results indicate that the violation of causal inequalities is rather ubiquitous in this framework, as for every nontrivial inequality considered it was possible to find a corresponding violation. The physical interpretation of such causal-inequality-violating processes, however, remains unclear. As it remains uncertain whether violation of bipartite causal inequalities is possible with physically realizable resources, it may be useful to narrow down the search for physical violations of causal inequalities by further developing the notion of `genuine multipartite non-causality'. Such a notion would allow one to single out causal inequalities whose violation cannot be obtained, for instance, by bipartite noncausal processes interacting causally with a third party.

\begin{acknowledgments}
We acknowledge discussions with Mateus Araújo, Časlav Brukner, Adrien Feix, and Ognyan Oreshkov.
This work was supported in part by the French National Research Agency (`Retour Post-Doctorants' program ANR-13-PDOC-0026), the European Commission (Marie Curie International Incoming Fellowship PIIF-GA-2013-623456), the ARC Centres for Engineered Quantum Systems (CE110001013) and for Quantum Computation and Communication Technology (CE110001027), and the Templeton World Charity Foundation (TWCF 0064/AB38).
\end{acknowledgments}

\appendix
\renewcommand{\thesection}{}
\renewcommand{\theequation}{A\arabic{equation}}
\setcounter{equation}{0}

\medskip

\section{Proofs of the claims of Subsection~\ref{subsec:a_few_comments}}

In this Appendix we prove the claims of Subsection~\ref{subsec:a_few_comments}. Just as for the definition of multipartite causal correlations, the proofs (except for the last one) are inductive.

\subsubsection{Convexity of causal correlations}

For $N=1$, it is obvious that (unconstrained) single-party correlations form a convex set.

\medskip

Let us assume that for $N \ge 2$, the set of $(N{-}1)$-partite causal correlations is convex, and let $P'$ and $P''$ be two $N$-partite causal correlations:
\ba
P'(\vec a|\vec x) &=& \sum_{k \in {\cal N}} \ q_k' \ P_k'(a_k|x_k) \ P_{k,x_k,a_k}'(\vec a_{\backslash k}|\vec x_{\backslash k}) \,, \\
P''(\vec a|\vec x) &=& \sum_{k \in {\cal N}} \ q_k'' \ P_k''(a_k|x_k) \ P_{k,x_k,a_k}''(\vec a_{\backslash k}|\vec x_{\backslash k}) \,, \qquad
\ea
with $q_k', q_k'' \ge 0$, $\sum_k q_k' = \sum_k q_k'' = 1$, with $P_k'$ and $P_k''$ being two valid probability distributions, and with $P_{k,x_k,a_k}'$ and $P_{k,x_k,a_k}''$ being $(N{-}1)$-partite causal correlations.

Let $q', q'' \geq 0$ with $q' + q'' = 1$, and define $P = q' P' + q'' P''$. Then one has
\ba
P(\vec a|\vec x) &=& q' P'(\vec a|\vec x) + q'' P''(\vec a|\vec x) \nonumber \\
&=& \sum_{k \in {\cal N}} q_k \, P_k(a_k|x_k) \ \Big( r_{x_k,a_k}' \, P_{k,x_k,a_k}'(\vec a_{\backslash k}|\vec x_{\backslash k}) \nonumber \\[-3mm]
&& \hspace{23mm} + \, r_{x_k,a_k}'' \, P_{k,x_k,a_k}''(\vec a_{\backslash k}|\vec x_{\backslash k}) \Big) , \ \qquad \label{eq:proof_convexity}
\ea
with $q_k = q' q_k' + q'' q_k''$, $q_k \, P_k(a_k|x_k) = q' q_k' \, P_k'(a_k|x_k) + q'' q_k'' \, P_k''(a_k|x_k)$, $r_{x_k,a_k}^{\prime(\prime)} = \frac{q^{\prime(\prime)} q_k^{\prime(\prime)} \, P_k^{\prime(\prime)}(a_k|x_k)}{q_k \, P_k(a_k|x_k)}$.
These are such that $q_k \ge 0$, $\sum_k q_k = 1$, $P_k(a_k|x_k)$ is a valid single-party probability distribution, $r_{x_k,a_k}^{\prime(\prime)}\ge 0$ and $r_{x_k,a_k}' + r_{x_k,a_k}'' = 1$ for all $x_k,a_k$. By the convexity of $(N{-}1)$-partite causal correlations, $P_{k,x_k,a_k}(\vec a_{\backslash k}|\vec x_{\backslash k}) := r_{x_k,a_k}' \, P_{k,x_k,a_k}'(\vec a_{\backslash k}|\vec x_{\backslash k}) + \, r_{x_k,a_k}'' \, P_{k,x_k,a_k}''(\vec a_{\backslash k}|\vec x_{\backslash k})$ is causal. Hence, Eq.~\eqref{eq:proof_convexity} is of the form~\eqref{def:causal_correlation}, so that $P = q' P' + q'' P''$ is causal, which by induction concludes the proof.
\qed

\subsubsection{Ignoring certain parties}

The claim, that any marginal correlation (for any nonempty subset ${\cal K}$ of parties) of an $N$-partite causal correlation is causal, is trivial to verify for the cases $N = 1$ and $2$, or when $|{\cal K}| = 1$ (recalling in particular that any single-partite correlation is causal).

\medskip

Suppose, for $N \ge 2$, that the claim holds for all marginals of all $(N{-}1)$-partite causal correlations, and consider an $N$-partite causal correlation $P(\vec a|\vec x)$, decomposed as in~\eqref{def:causal_correlation}, and a subset $\cal K$ of ${\cal N} = \{1, \ldots, N\}$, with $|{\cal K}| \ge 2$.

Let us first take each term
\ba
P_k(\vec a|\vec x) := \pr{k}{k}\, \prv{k,x_k,a_k}{\backslash k}
\ea
in the sum of~\eqref{def:causal_correlation} separately.
If $k \in {\cal K}$, then
\ba
\sum_{\vec a_{{\cal N} \backslash {\cal K}}} P_k(\vec a|\vec x) = \pr{k}{k}\, \sum_{\vec a_{{\cal N} \backslash {\cal K}}} \prv{k,x_k,a_k}{\backslash k} \,. \qquad \label{eq:proof_marginal_in_K}
\ea
By the induction hypothesis, the marginal 
$(|{\cal K}|{-}1)$-partite 
correlations $\sum_{\vec a_{{\cal N} \backslash {\cal K}}} \prv{k,x_k,a_k}{\backslash k}$ of the $(N{-}1)$-partite causal correlations $P_{k,x_k,a_k}$ are causal;
Eq.~\eqref{eq:proof_marginal_in_K} is thus of the form~\eqref{def:causal_correlation}, so that the marginal $|{\cal K}|$-partite correlation $\sum_{\vec a_{{\cal N} \backslash {\cal K}}} P_k(\vec a|\vec x)$ is itself causal.
If $k \in {\cal N} \backslash {\cal K}$ on the other hand, then
\ba
\sum_{\vec a_{{\cal N} \backslash {\cal K}}} \! P_k(\vec a|\vec x) = \sum_{a_k} \pr{k}{k}\!\! \sum_{\vec a_{{\cal N} \backslash {\cal K} \backslash k}} \!\!\! \prv{k,x_k,a_k}{\backslash k} \,. \qquad \label{eq:proof_marginal_in_N_K}
\ea
Again by the induction hypothesis, the marginal 
$|{\cal K}|$-partite 
correlations $\sum_{\vec a_{{\cal N} \backslash {\cal K} \backslash k}} \prv{k,x_k,a_k}{\backslash k}$ of the $(N{-}1)$-partite causal correlation $P_{k,x_k,a_k}$ are causal; Eq.~\eqref{eq:proof_marginal_in_N_K} then defines a convex combination of causal correlations, with normalized weights $\pr{k}{k}$, so that the marginal correlation $\sum_{\vec a_{{\cal N} \backslash {\cal K}}} P_k(\vec a|\vec x)$ is again causal.

Coming back to the sum of Eq.~\eqref{def:causal_correlation}, we thus find that the marginal $|{\cal K}|$-partite correlation
\ba
\sum_{\vec a_{{\cal N} \backslash {\cal K}}}\prv{}{} = \sum_{k \in {\cal N}} q_k\ \Big( \sum_{\vec a_{{\cal N} \backslash {\cal K}}} P_k(\vec a|\vec x) \Big)
\ea
is a convex combination of causal correlations, which implies that it is itself causal. By induction this concludes the proof of our claim on marginal causal correlations.
\qed

\subsubsection{Combining causal correlations `one after the other'}

Consider two nonempty sets of parties ${\cal K}$ and ${\cal N} \backslash {\cal K}$, and causal correlations $P(\vec a_{\cal K}|\vec x_{\cal K})$ and $P_{{\cal K},\vec x_{\cal K},\vec a_{\cal K}}(\vec a_{{\cal N} \backslash {\cal K}}|\vec x_{{\cal N} \backslash {\cal K}})$ for each set, respectively, and define
\ba
P(\vec a|\vec x) := P(\vec a_{\cal K}|\vec x_{\cal K}) \ P_{{\cal K},\vec x_{\cal K},\vec a_{\cal K}}(\vec a_{{\cal N} \backslash {\cal K}}|\vec x_{{\cal N} \backslash {\cal K}}) \,. ~\label{eq:combine_causal_proof}
\ea

\medskip

If the first set contains only $|{\cal K}| = 1$ party, then for any complementary subset ${\cal N} \backslash {\cal K}$ (that is, for any total number of parties $N$) the combined correlation $P(\vec a|\vec x)$ is clearly of the form~\eqref{def:causal_correlation}, and is therefore causal.

\medskip

For $|{\cal K}| \ge 2$, assume that for all subsets ${\cal K}'$ with $|{\cal K}'| = |{\cal K}|-1$ and all complementary subsets ${\cal N}' \backslash {\cal K}'$ (for any ${\cal N}' = \{1,\ldots,N'\}$, $N' > |{\cal K}'|$), the combination of two causal correlations as in~\eqref{eq:combine_causal_proof} yields another causal correlation. By assumption, the causal correlation $P(\vec a_{\cal K}|\vec x_{\cal K})$ under consideration can be decomposed in the form~\eqref{def:causal_correlation} as
\ba
P(\vec a_{\cal K}|\vec x_{\cal K}) = \sum_{k \in {\cal K}} q_k\ \pr{k}{k}\ \prv{k,x_k,a_k}{{\cal K}\backslash k}, \notag \\[-3mm]
\ea
where for each $k, x_k,a_k$, the $(|{\cal K}|{-}1)$-partite correlation $\prv{k,x_k,a_k}{{\cal K}\backslash k}$ is causal. We can then write
\ba
P(\vec a|\vec x) = \sum_{k \in {\cal K}} q_k\ \pr{k}{k}\ \prv{k,x_k,a_k}{\backslash k} \,, \qquad \label{eq:combine_causal_proof2}
\ea
where we defined
\begin{align}
& \hspace{-3mm} \prv{k,x_k,a_k}{\backslash k} \notag \\
& = \prv{k,x_k,a_k}{{\cal K}\backslash k} \ P_{{\cal K},\vec x_{\cal K},\vec a_{\cal K}}(\vec a_{{\cal N} \backslash {\cal K}}|\vec x_{{\cal N} \backslash {\cal K}})\, ,
\end{align}
which for each $k,x_k,a_k$, is the combination of a $(|{\cal K}|{-}1)$-partite causal correlation with an $(N{-}|{\cal K}|)$-partite causal correlation, as in~\eqref{eq:combine_causal_proof}. By the induction hypothesis, it is therefore causal. We thus find that the correlation $P(\vec a|\vec x)$ written as in~\eqref{eq:combine_causal_proof2} is of the form~\eqref{def:causal_correlation} and hence causal, which by induction concludes the proof.
\qed

\subsubsection{An equivalent characterization of causal correlations}

Clearly, Eq.~\eqref{def:causal_correlation} is a particular case of~\eqref{def:causal_correlation_v2}, so that any causal correlation (of the form~\eqref{def:causal_correlation} by definition) is also of the form~\eqref{def:causal_correlation_v2}. It remains to be shown that any correlation of the form~\eqref{def:causal_correlation_v2} is also causal.

\medskip

Consider such an $N$-partite correlation $P(\vec a|\vec x)$ of the form~\eqref{def:causal_correlation_v2}, with all $P_{\cal K}(\vec a_{\cal K}|\vec x_{\cal K})$ and $P_{{\cal K},\vec x_{\cal K},\vec a_{\cal K}}(\vec a_{{\cal N} \backslash {\cal K}}|\vec x_{{\cal N} \backslash {\cal K}})$ being causal correlations.
Consider for now one term in the sum~\eqref{def:causal_correlation_v2}, corresponding to one specific subset ${\cal K}$. 
If $|{\cal K}|=1$, then $P_{\cal K}(\vec a_{\cal K}|\vec x_{\cal K}) \ P_{{\cal K},\vec x_{\cal K},\vec a_{\cal K}}(\vec a_{{\cal N} \backslash {\cal K}}|\vec x_{{\cal N} \backslash {\cal K}})$ is directly of the form~\eqref{def:causal_correlation}, and hence it is causal. Otherwise, if $|{\cal K}| \ge 2$ then by the definition of causal correlations, $P_{\cal K}(\vec a_{\cal K}|\vec x_{\cal K})$ can be decomposed in the form~\eqref{def:causal_correlation}, as
\ba
P_{\cal K}(\vec a_{\cal K}|\vec x_{\cal K}) = \sum_{k \in {\cal K}} q_{k}\ P_{k}(a_{k}|x_{k})\ P_{k,x_k,a_k}(\vec a_{{\cal K}\backslash {k}}|\vec x_{{\cal K}\backslash {k}}) \nonumber \\[-3mm]
\ea
with $P_{k,x_k,a_k}(\vec a_{{\cal K}\backslash {k}}|\vec x_{{\cal K}\backslash {k}})$ causal, so that
\ba
&& P_{\cal K}(\vec a_{\cal K}|\vec x_{\cal K}) \ P_{{\cal K},\vec x_{\cal K},\vec a_{\cal K}}(\vec a_{{\cal N} \backslash {\cal K}}|\vec x_{{\cal N} \backslash {\cal K}}) \nonumber \\
&& \quad = \sum_{k \in {\cal K}} q_{k}\ P_{k}(a_{k}|x_{k})\ P_{k,x_k,a_k}(\vec a_{\backslash k}|\vec x_{\backslash k}) \,, \quad \label{eq:proof_equiv_form}
\ea
where we defined
\ba
&& P_{k,x_k,a_k}(\vec a_{\backslash k}|\vec x_{\backslash k}) \nonumber \\
&& \quad = P_{k,x_k,a_k}(\vec a_{{\cal K}\backslash {k}}|\vec x_{{\cal K}\backslash {k}}) \ P_{{\cal K},\vec x_{\cal K},\vec a_{\cal K}}(\vec a_{{\cal N} \backslash {\cal K}}|\vec x_{{\cal N} \backslash {\cal K}}) \,. \qquad
\ea
As a product of causal correlations, the correlation $P_{k,x_k,a_k}(\vec a_{\backslash k}|\vec x_{\backslash k})$ thus defined is causal (see the previous subsection). Hence, the product $P_{\cal K}(\vec a_{\cal K}|\vec x_{\cal K}) \ P_{{\cal K},\vec x_{\cal K},\vec a_{\cal K}}(\vec a_{{\cal N} \backslash {\cal K}}|\vec x_{{\cal N} \backslash {\cal K}})$ from Eq.~\eqref{eq:proof_equiv_form} is of the form~\eqref{def:causal_correlation}, and therefore it is causal.

\medskip

Put together, we find that a correlation of the form~\eqref{def:causal_correlation_v2} can be expressed as a convex combination of causal correlations (for each ${\cal K}$), and is therefore causal, which concludes the proof.
\qed

\bibliography{CausalCorrelations}

\begin{thebibliography}{33}%
\makeatletter
\providecommand \@ifxundefined [1]{%
 \@ifx{#1\undefined}
}%
\providecommand \@ifnum [1]{%
 \ifnum #1\expandafter \@firstoftwo
 \else \expandafter \@secondoftwo
 \fi
}%
\providecommand \@ifx [1]{%
 \ifx #1\expandafter \@firstoftwo
 \else \expandafter \@secondoftwo
 \fi
}%
\providecommand \natexlab [1]{#1}%
\providecommand \enquote  [1]{``#1''}%
\providecommand \bibnamefont  [1]{#1}%
\providecommand \bibfnamefont [1]{#1}%
\providecommand \citenamefont [1]{#1}%
\providecommand \href@noop [0]{\@secondoftwo}%
\providecommand \href [0]{\begingroup \@sanitize@url \@href}%
\providecommand \@href[1]{\@@startlink{#1}\@@href}%
\providecommand \@@href[1]{\endgroup#1\@@endlink}%
\providecommand \@sanitize@url [0]{\catcode `\\12\catcode `\$12\catcode
  `\&12\catcode `\#12\catcode `\^12\catcode `\_12\catcode `\%12\relax}%
\providecommand \@@startlink[1]{}%
\providecommand \@@endlink[0]{}%
\providecommand \url  [0]{\begingroup\@sanitize@url \@url }%
\providecommand \@url [1]{\endgroup\@href {#1}{\urlprefix }}%
\providecommand \urlprefix  [0]{URL }%
\providecommand \Eprint [0]{\href }%
\providecommand \doibase [0]{http://dx.doi.org/}%
\providecommand \selectlanguage [0]{\@gobble}%
\providecommand \bibinfo  [0]{\@secondoftwo}%
\providecommand \bibfield  [0]{\@secondoftwo}%
\providecommand \translation [1]{[#1]}%
\providecommand \BibitemOpen [0]{}%
\providecommand \bibitemStop [0]{}%
\providecommand \bibitemNoStop [0]{.\EOS\space}%
\providecommand \EOS [0]{\spacefactor3000\relax}%
\providecommand \BibitemShut  [1]{\csname bibitem#1\endcsname}%
\let\auto@bib@innerbib\@empty
%</preamble>
\bibitem [{\citenamefont {Bell}(1964)}]{bell64}%
  \BibitemOpen
  \bibfield  {author} {\bibinfo {author} {\bibfnamefont {J.~S.}\ \bibnamefont
  {Bell}},\ }\href@noop {} {\bibfield  {journal} {\bibinfo  {journal}
  {Physics}\ }\textbf {\bibinfo {volume} {1}},\ \bibinfo {pages} {195}
  (\bibinfo {year} {1964})}\BibitemShut {NoStop}%
\bibitem [{\citenamefont {Branciard}\ \emph {et~al.}(2012)\citenamefont
  {Branciard}, \citenamefont {Rosset}, \citenamefont {Gisin},\ and\
  \citenamefont {Pironio}}]{BranciardBilocal2012}%
  \BibitemOpen
  \bibfield  {author} {\bibinfo {author} {\bibfnamefont {C.}~\bibnamefont
  {Branciard}}, \bibinfo {author} {\bibfnamefont {D.}~\bibnamefont {Rosset}},
  \bibinfo {author} {\bibfnamefont {N.}~\bibnamefont {Gisin}}, \ and\ \bibinfo
  {author} {\bibfnamefont {S.}~\bibnamefont {Pironio}},\ }\href {\doibase
  10.1103/PhysRevA.85.032119} {\bibfield  {journal} {\bibinfo  {journal} {Phys.
  Rev. A}\ }\textbf {\bibinfo {volume} {85}},\ \bibinfo {pages} {032119}
  (\bibinfo {year} {2012})},\ \Eprint {http://arxiv.org/abs/1112.4502}
  {arXiv:1112.4502 [quant-ph]} \BibitemShut {NoStop}%
\bibitem [{\citenamefont {Fritz}(2012)}]{FritzBeyond2012}%
  \BibitemOpen
  \bibfield  {author} {\bibinfo {author} {\bibfnamefont {T.}~\bibnamefont
  {Fritz}},\ }\href {\doibase 10.1088/1367-2630/14/10/103001} {\bibfield
  {journal} {\bibinfo  {journal} {New J. Phys.}\ }\textbf {\bibinfo {volume}
  {14}},\ \bibinfo {pages} {103001} (\bibinfo {year} {2012})},\ \Eprint
  {http://arxiv.org/abs/1206.5115} {arXiv:1206.5115 [quant-ph]} \BibitemShut
  {NoStop}%
\bibitem [{\citenamefont {Chaves}\ \emph {et~al.}(2014)\citenamefont {Chaves},
  \citenamefont {Luft},\ and\ \citenamefont {Gross}}]{ChavesCausal2014}%
  \BibitemOpen
  \bibfield  {author} {\bibinfo {author} {\bibfnamefont {R.}~\bibnamefont
  {Chaves}}, \bibinfo {author} {\bibfnamefont {L.}~\bibnamefont {Luft}}, \ and\
  \bibinfo {author} {\bibfnamefont {D.}~\bibnamefont {Gross}},\ }\href
  {\doibase 10.1088/1367-2630/16/4/043001} {\bibfield  {journal} {\bibinfo
  {journal} {New J. Phys.}\ }\textbf {\bibinfo {volume} {16}},\ \bibinfo
  {pages} {043001} (\bibinfo {year} {2014})},\ \Eprint
  {http://arxiv.org/abs/1310.0284} {arXiv:1310.0284 [quant-ph]} \BibitemShut
  {NoStop}%
\bibitem [{\citenamefont {Fritz}(2016)}]{fritzbeyond2015}%
  \BibitemOpen
  \bibfield  {author} {\bibinfo {author} {\bibfnamefont {T.}~\bibnamefont
  {Fritz}},\ }\href {\doibase 10.1007/s00220-015-2495-5} {\bibfield  {journal}
  {\bibinfo  {journal} {Comm. Math. Phys.}\ }\textbf {\bibinfo {volume}
  {341}},\ \bibinfo {pages} {391} (\bibinfo {year} {2016})},\ \Eprint
  {http://arxiv.org/abs/1404.4812} {arXiv:1404.4812 [quant-ph]} \BibitemShut
  {NoStop}%
\bibitem [{\citenamefont {Henson}\ \emph {et~al.}(2014)\citenamefont {Henson},
  \citenamefont {Lal},\ and\ \citenamefont {Pusey}}]{Henson2014}%
  \BibitemOpen
  \bibfield  {author} {\bibinfo {author} {\bibfnamefont {J.}~\bibnamefont
  {Henson}}, \bibinfo {author} {\bibfnamefont {R.}~\bibnamefont {Lal}}, \ and\
  \bibinfo {author} {\bibfnamefont {M.~F.}\ \bibnamefont {Pusey}},\ }\href
  {\doibase 10.1088/1367-2630/16/11/113043} {\bibfield  {journal} {\bibinfo
  {journal} {New. J.~Phys.}\ }\textbf {\bibinfo {volume} {16}},\ \bibinfo
  {pages} {113043} (\bibinfo {year} {2014})},\ \Eprint
  {http://arxiv.org/abs/1405.2572} {arXiv:1405.2572 [quant-ph]} \BibitemShut
  {NoStop}%
\bibitem [{\citenamefont {Pienaar}(2016)}]{pienaar2016causal}%
  \BibitemOpen
  \bibfield  {author} {\bibinfo {author} {\bibfnamefont {J.}~\bibnamefont
  {Pienaar}},\ }\href@noop {} {\  (\bibinfo {year} {2016})},\ \Eprint
  {http://arxiv.org/abs/1606.07798} {arXiv:1606.07798 [quant-ph]} \BibitemShut
  {NoStop}%
\bibitem [{\citenamefont {Oreshkov}\ \emph {et~al.}(2012)\citenamefont
  {Oreshkov}, \citenamefont {Costa},\ and\ \citenamefont
  {Brukner}}]{oreshkov12}%
  \BibitemOpen
  \bibfield  {author} {\bibinfo {author} {\bibfnamefont {O.}~\bibnamefont
  {Oreshkov}}, \bibinfo {author} {\bibfnamefont {F.}~\bibnamefont {Costa}}, \
  and\ \bibinfo {author} {\bibfnamefont {{\v{C}}.}~\bibnamefont {Brukner}},\
  }\href {\doibase 10.1038/ncomms2076} {\bibfield  {journal} {\bibinfo
  {journal} {Nat. Commun.}\ }\textbf {\bibinfo {volume} {3}},\ \bibinfo {pages}
  {1092} (\bibinfo {year} {2012})},\ \Eprint {http://arxiv.org/abs/1105.4464}
  {arXiv:1105.4464 [quant-ph]} \BibitemShut {NoStop}%
\bibitem [{\citenamefont {Chiribella}\ \emph {et~al.}(2013)\citenamefont
  {Chiribella}, \citenamefont {D'Ariano}, \citenamefont {Perinotti},\ and\
  \citenamefont {Valiron}}]{Chiribella:2013aa}%
  \BibitemOpen
  \bibfield  {author} {\bibinfo {author} {\bibfnamefont {G.}~\bibnamefont
  {Chiribella}}, \bibinfo {author} {\bibfnamefont {G.~M.}\ \bibnamefont
  {D'Ariano}}, \bibinfo {author} {\bibfnamefont {P.}~\bibnamefont {Perinotti}},
  \ and\ \bibinfo {author} {\bibfnamefont {B.}~\bibnamefont {Valiron}},\ }\href
  {\doibase 10.1103/PhysRevA.88.022318} {\bibfield  {journal} {\bibinfo
  {journal} {Phys. Rev. A}\ }\textbf {\bibinfo {volume} {88}},\ \bibinfo
  {pages} {022318} (\bibinfo {year} {2013})},\ \Eprint
  {http://arxiv.org/abs/0912.0195} {arXiv:0912.0195 [quant-ph]} \BibitemShut
  {NoStop}%
\bibitem [{\citenamefont {{Ara{\'u}jo}}\ \emph {et~al.}(2014)\citenamefont
  {{Ara{\'u}jo}}, \citenamefont {{Costa}},\ and\ \citenamefont
  {{Brukner}}}]{araujo14}%
  \BibitemOpen
  \bibfield  {author} {\bibinfo {author} {\bibfnamefont {M.}~\bibnamefont
  {{Ara{\'u}jo}}}, \bibinfo {author} {\bibfnamefont {F.}~\bibnamefont
  {{Costa}}}, \ and\ \bibinfo {author} {\bibfnamefont {{\v C}.}~\bibnamefont
  {{Brukner}}},\ }\href {\doibase 10.1103/PhysRevLett.113.250402} {\bibfield
  {journal} {\bibinfo  {journal} {Phys. Rev. Lett.}\ }\textbf {\bibinfo
  {volume} {113}},\ \bibinfo {pages} {250402} (\bibinfo {year} {2014})},\
  \Eprint {http://arxiv.org/abs/1401.8127} {arXiv:1401.8127 [quant-ph]}
  \BibitemShut {NoStop}%
\bibitem [{\citenamefont {Feix}\ \emph {et~al.}(2015)\citenamefont {Feix},
  \citenamefont {Ara\'ujo},\ and\ \citenamefont {Brukner}}]{feixquantum2015}%
  \BibitemOpen
  \bibfield  {author} {\bibinfo {author} {\bibfnamefont {A.}~\bibnamefont
  {Feix}}, \bibinfo {author} {\bibfnamefont {M.}~\bibnamefont {Ara\'ujo}}, \
  and\ \bibinfo {author} {\bibfnamefont {{\v C}.}~\bibnamefont {Brukner}},\
  }\href {\doibase 10.1103/PhysRevA.92.052326} {\bibfield  {journal} {\bibinfo
  {journal} {Phys. Rev. A}\ }\textbf {\bibinfo {volume} {92}},\ \bibinfo
  {pages} {052326} (\bibinfo {year} {2015})},\ \Eprint
  {http://arxiv.org/abs/1508.07840} {arXiv:1508.07840 [quant-ph]} \BibitemShut
  {NoStop}%
\bibitem [{\citenamefont {Gu{\'e}rin}\ \emph {et~al.}(2016)\citenamefont
  {Gu{\'e}rin}, \citenamefont {Feix}, \citenamefont {Ara{\'u}jo},\ and\
  \citenamefont {Brukner}}]{guerin16}%
  \BibitemOpen
  \bibfield  {author} {\bibinfo {author} {\bibfnamefont {P.~A.}\ \bibnamefont
  {Gu{\'e}rin}}, \bibinfo {author} {\bibfnamefont {A.}~\bibnamefont {Feix}},
  \bibinfo {author} {\bibfnamefont {M.}~\bibnamefont {Ara{\'u}jo}}, \ and\
  \bibinfo {author} {\bibfnamefont {{\v C}.}~\bibnamefont {Brukner}},\ }\href
  {\doibase 10.1103/PhysRevLett.117.100502} {\bibfield  {journal} {\bibinfo
  {journal} {Phys. Rev. Lett.}\ }\textbf {\bibinfo {volume} {117}},\ \bibinfo
  {pages} {100502} (\bibinfo {year} {2016})},\ \Eprint
  {http://arxiv.org/abs/1605.07372} {arXiv:1605.07372 [quant-ph]} \BibitemShut
  {NoStop}%
\bibitem [{\citenamefont {Procopio}\ \emph {et~al.}(2015)\citenamefont
  {Procopio}, \citenamefont {Moqanaki}, \citenamefont {Ara\'{u}jo},
  \citenamefont {Costa}, \citenamefont {Alonso~Calafell}, \citenamefont {Dowd},
  \citenamefont {Hamel}, \citenamefont {Rozema}, \citenamefont {Brukner},\ and\
  \citenamefont {Walther}}]{Procopio:2015aa}%
  \BibitemOpen
  \bibfield  {author} {\bibinfo {author} {\bibfnamefont {L.~M.}\ \bibnamefont
  {Procopio}}, \bibinfo {author} {\bibfnamefont {A.}~\bibnamefont {Moqanaki}},
  \bibinfo {author} {\bibfnamefont {M.}~\bibnamefont {Ara\'{u}jo}}, \bibinfo
  {author} {\bibfnamefont {F.}~\bibnamefont {Costa}}, \bibinfo {author}
  {\bibfnamefont {I.}~\bibnamefont {Alonso~Calafell}}, \bibinfo {author}
  {\bibfnamefont {E.~G.}\ \bibnamefont {Dowd}}, \bibinfo {author}
  {\bibfnamefont {D.~R.}\ \bibnamefont {Hamel}}, \bibinfo {author}
  {\bibfnamefont {L.~A.}\ \bibnamefont {Rozema}}, \bibinfo {author}
  {\bibfnamefont {{\v C}.}~\bibnamefont {Brukner}}, \ and\ \bibinfo {author}
  {\bibfnamefont {P.}~\bibnamefont {Walther}},\ }\href {\doibase
  10.1038/ncomms8913} {\bibfield  {journal} {\bibinfo  {journal} {Nat.
  Commun.}\ }\textbf {\bibinfo {volume} {6}},\ \bibinfo {pages} {7913}
  (\bibinfo {year} {2015})},\ \Eprint {http://arxiv.org/abs/1412.4006}
  {arXiv:1412.4006 [quant-ph]} \BibitemShut {NoStop}%
\bibitem [{\citenamefont {Ara\'{u}jo}\ \emph {et~al.}(2015)\citenamefont
  {Ara\'{u}jo}, \citenamefont {Branciard}, \citenamefont {Costa}, \citenamefont
  {Feix}, \citenamefont {Giarmatzi},\ and\ \citenamefont {Brukner}}]{araujo15}%
  \BibitemOpen
  \bibfield  {author} {\bibinfo {author} {\bibfnamefont {M.}~\bibnamefont
  {Ara\'{u}jo}}, \bibinfo {author} {\bibfnamefont {C.}~\bibnamefont
  {Branciard}}, \bibinfo {author} {\bibfnamefont {F.}~\bibnamefont {Costa}},
  \bibinfo {author} {\bibfnamefont {A.}~\bibnamefont {Feix}}, \bibinfo {author}
  {\bibfnamefont {C.}~\bibnamefont {Giarmatzi}}, \ and\ \bibinfo {author}
  {\bibfnamefont {{\v{C}}.}~\bibnamefont {Brukner}},\ }\href {\doibase
  10.1088/1367-2630/17/10/102001} {\bibfield  {journal} {\bibinfo  {journal}
  {New. J.~Phys.}\ }\textbf {\bibinfo {volume} {17}},\ \bibinfo {pages}
  {102001} (\bibinfo {year} {2015})},\ \Eprint
  {http://arxiv.org/abs/1506.03776} {arXiv:1506.03776 [quant-ph]} \BibitemShut
  {NoStop}%
\bibitem [{\citenamefont {Oreshkov}\ and\ \citenamefont
  {Giarmatzi}(2016)}]{oreshkov15}%
  \BibitemOpen
  \bibfield  {author} {\bibinfo {author} {\bibfnamefont {O.}~\bibnamefont
  {Oreshkov}}\ and\ \bibinfo {author} {\bibfnamefont {C.}~\bibnamefont
  {Giarmatzi}},\ }\href {\doibase 10.1088/1367-2630/18/9/093020} {\bibfield
  {journal} {\bibinfo  {journal} {New J. Phys.}\ }\textbf {\bibinfo {volume}
  {18}},\ \bibinfo {pages} {093020} (\bibinfo {year} {2016})},\ \Eprint
  {http://arxiv.org/abs/1506.05449} {arXiv:1506.05449 [quant-ph]} \BibitemShut
  {NoStop}%
\bibitem [{\citenamefont {{Brukner}}(2014)}]{brukner14}%
  \BibitemOpen
  \bibfield  {author} {\bibinfo {author} {\bibfnamefont {{\v C}.}~\bibnamefont
  {{Brukner}}},\ }\href {\doibase 10.1038/nphys2930} {\bibfield  {journal}
  {\bibinfo  {journal} {Nat. Phys.}\ }\textbf {\bibinfo {volume} {10}},\
  \bibinfo {pages} {259} (\bibinfo {year} {2014})}\BibitemShut {NoStop}%
\bibitem [{\citenamefont {Branciard}\ \emph {et~al.}(2016)\citenamefont
  {Branciard}, \citenamefont {Ara\'{u}jo}, \citenamefont {Feix}, \citenamefont
  {Costa},\ and\ \citenamefont {Brukner}}]{branciard16}%
  \BibitemOpen
  \bibfield  {author} {\bibinfo {author} {\bibfnamefont {C.}~\bibnamefont
  {Branciard}}, \bibinfo {author} {\bibfnamefont {M.}~\bibnamefont
  {Ara\'{u}jo}}, \bibinfo {author} {\bibfnamefont {A.}~\bibnamefont {Feix}},
  \bibinfo {author} {\bibfnamefont {F.}~\bibnamefont {Costa}}, \ and\ \bibinfo
  {author} {\bibfnamefont {{\v{C}}.}~\bibnamefont {Brukner}},\ }\href {\doibase
  10.1088/1367-2630/18/1/013008} {\bibfield  {journal} {\bibinfo  {journal}
  {New. J.~Phys.}\ }\textbf {\bibinfo {volume} {18}},\ \bibinfo {pages}
  {013008} (\bibinfo {year} {2016})},\ \Eprint
  {http://arxiv.org/abs/1508.01704} {arXiv:1508.01704 [quant-ph]} \BibitemShut
  {NoStop}%
\bibitem [{\citenamefont {Hardy}(2005)}]{hardy2005probability}%
  \BibitemOpen
  \bibfield  {author} {\bibinfo {author} {\bibfnamefont {L.}~\bibnamefont
  {Hardy}},\ }\href@noop {} {\  (\bibinfo {year} {2005})},\ \Eprint
  {http://arxiv.org/abs/gr-qc/0509120} {arXiv:gr-qc/0509120} \BibitemShut
  {NoStop}%
\bibitem [{\citenamefont {Baumeler}\ and\ \citenamefont
  {Wolf}(2014)}]{baumeler13}%
  \BibitemOpen
  \bibfield  {author} {\bibinfo {author} {\bibfnamefont {{\"{A}}.}~\bibnamefont
  {Baumeler}}\ and\ \bibinfo {author} {\bibfnamefont {S.}~\bibnamefont
  {Wolf}},\ }in\ \href {\doibase 10.1109/ISIT.2014.6874888} {\emph {\bibinfo
  {booktitle} {2014 IEEE International Symposium on Information Theory
  (ISIT)}}}\ (\bibinfo  {publisher} {IEEE},\ \bibinfo {address} {Piscataway,
  NJ},\ \bibinfo {year} {2014})\ pp.\ \bibinfo {pages} {526--530},\ \Eprint
  {http://arxiv.org/abs/1312.5916} {arXiv:1312.5916 [quant-ph]} \BibitemShut
  {NoStop}%
\bibitem [{\citenamefont {Fine}(1982)}]{fine82}%
  \BibitemOpen
  \bibfield  {author} {\bibinfo {author} {\bibfnamefont {A.}~\bibnamefont
  {Fine}},\ }\href {\doibase 10.1103/PhysRevLett.48.291} {\bibfield  {journal}
  {\bibinfo  {journal} {Phys. Rev. Lett.}\ }\textbf {\bibinfo {volume} {48}},\
  \bibinfo {pages} {291} (\bibinfo {year} {1982})}\BibitemShut {NoStop}%
\bibitem [{\citenamefont {Pironio}(2005)}]{pironio05}%
  \BibitemOpen
  \bibfield  {author} {\bibinfo {author} {\bibfnamefont {S.}~\bibnamefont
  {Pironio}},\ }\href {\doibase 10.1063/1.1928727} {\bibfield  {journal}
  {\bibinfo  {journal} {J.~Math. Phys.}\ }\textbf {\bibinfo {volume} {46}},\
  \bibinfo {pages} {062112} (\bibinfo {year} {2005})},\ \Eprint
  {http://arxiv.org/abs/quant-ph/0503179} {arXiv:quant-ph/0503179} \BibitemShut
  {NoStop}%
\bibitem [{\citenamefont {Branciard}(2016)}]{branciard16b}%
  \BibitemOpen
  \bibfield  {author} {\bibinfo {author} {\bibfnamefont {C.}~\bibnamefont
  {Branciard}},\ }\href {\doibase 10.1038/srep26018} {\bibfield  {journal}
  {\bibinfo  {journal} {Sci. Rep.}\ }\textbf {\bibinfo {volume} {6}},\ \bibinfo
  {pages} {26018} (\bibinfo {year} {2016})},\ \Eprint
  {http://arxiv.org/abs/1603.00043} {arXiv:1603.00043 [quant-ph]} \BibitemShut
  {NoStop}%
\bibitem [{\citenamefont {Fukuda}(2012)}]{cdd}%
  \BibitemOpen
  \bibfield  {author} {\bibinfo {author} {\bibfnamefont {K.}~\bibnamefont
  {Fukuda}},\ }\href@noop {} {\enquote {\bibinfo {title}
  {\href{https://www.inf.ethz.ch/personal/fukudak/cdd_home/}{\textsc{cdd},
  v0.94g}},}\ } (\bibinfo {year} {2012}),\ \bibinfo {note}
  {\url{https://www.inf.ethz.ch/personal/fukudak/cdd_home/}}\BibitemShut
  {NoStop}%
\bibitem [{\citenamefont {Ara\'{u}jo}\ and\ \citenamefont
  {Feix}(2016)}]{AraujoFeixIneq}%
  \BibitemOpen
  \bibfield  {author} {\bibinfo {author} {\bibfnamefont {M.}~\bibnamefont
  {Ara\'{u}jo}}\ and\ \bibinfo {author} {\bibfnamefont {A.}~\bibnamefont
  {Feix}},\ }\href@noop {} {}\bibinfo {howpublished} {private communication}
  (\bibinfo {year} {2016})\BibitemShut {NoStop}%
\bibitem [{\citenamefont {Almeida}\ \emph {et~al.}(2010)\citenamefont
  {Almeida}, \citenamefont {Bancal}, \citenamefont {Brunner}, \citenamefont
  {Ac\'{i}n}, \citenamefont {Gisin},\ and\ \citenamefont
  {Pironio}}]{Almeida:2010aa}%
  \BibitemOpen
  \bibfield  {author} {\bibinfo {author} {\bibfnamefont {M.~L.}\ \bibnamefont
  {Almeida}}, \bibinfo {author} {\bibfnamefont {J.-D.}\ \bibnamefont {Bancal}},
  \bibinfo {author} {\bibfnamefont {N.}~\bibnamefont {Brunner}}, \bibinfo
  {author} {\bibfnamefont {A.}~\bibnamefont {Ac\'{i}n}}, \bibinfo {author}
  {\bibfnamefont {N.}~\bibnamefont {Gisin}}, \ and\ \bibinfo {author}
  {\bibfnamefont {S.}~\bibnamefont {Pironio}},\ }\href {\doibase
  10.1103/PhysRevLett.104.230404} {\bibfield  {journal} {\bibinfo  {journal}
  {Phys. Rev. Lett.}\ }\textbf {\bibinfo {volume} {104}},\ \bibinfo {pages}
  {230404} (\bibinfo {year} {2010})},\ \Eprint {http://arxiv.org/abs/1003.3844}
  {arXiv:1003.3844 [quant-ph]} \BibitemShut {NoStop}%
\bibitem [{\citenamefont {Davies}\ and\ \citenamefont
  {Lewis}(1970)}]{Davies:1970aa}%
  \BibitemOpen
  \bibfield  {author} {\bibinfo {author} {\bibfnamefont {E.~B.}\ \bibnamefont
  {Davies}}\ and\ \bibinfo {author} {\bibfnamefont {J.~T.}\ \bibnamefont
  {Lewis}},\ }\href {\doibase 10.1007/BF01647093} {\bibfield  {journal}
  {\bibinfo  {journal} {Commun. Math. Phys.}\ }\textbf {\bibinfo {volume}
  {17}},\ \bibinfo {pages} {239} (\bibinfo {year} {1970})}\BibitemShut
  {NoStop}%
\bibitem [{\citenamefont {Choi}(1975)}]{Choi:1975aa}%
  \BibitemOpen
  \bibfield  {author} {\bibinfo {author} {\bibfnamefont {M.-D.}\ \bibnamefont
  {Choi}},\ }\href {\doibase 10.1016/0024-3795(75)90075-0} {\bibfield
  {journal} {\bibinfo  {journal} {Linear Algebra Appl.}\ }\textbf {\bibinfo
  {volume} {10}},\ \bibinfo {pages} {285} (\bibinfo {year} {1975})}\BibitemShut
  {NoStop}%
\bibitem [{\citenamefont {Jamio\l{}kowski}(1972)}]{Jamiolkowski:1972aa}%
  \BibitemOpen
  \bibfield  {author} {\bibinfo {author} {\bibfnamefont {A.}~\bibnamefont
  {Jamio\l{}kowski}},\ }\href {\doibase 10.1016/0034-4877(72)90011-0}
  {\bibfield  {journal} {\bibinfo  {journal} {Rep. Math. Phys.}\ }\textbf
  {\bibinfo {volume} {3}},\ \bibinfo {pages} {275} (\bibinfo {year}
  {1972})}\BibitemShut {NoStop}%
\bibitem [{\citenamefont {Baumeler}\ and\ \citenamefont
  {Wolf}(2016)}]{Baumelerspace2016}%
  \BibitemOpen
  \bibfield  {author} {\bibinfo {author} {\bibfnamefont {{\"A}.}~\bibnamefont
  {Baumeler}}\ and\ \bibinfo {author} {\bibfnamefont {S.}~\bibnamefont
  {Wolf}},\ }\href {http://stacks.iop.org/1367-2630/18/i=1/a=013036} {\bibfield
   {journal} {\bibinfo  {journal} {New. J.~Phys.}\ }\textbf {\bibinfo {volume}
  {18}},\ \bibinfo {pages} {013036} (\bibinfo {year} {2016})},\ \Eprint
  {http://arxiv.org/abs/1507.01714} {arXiv:1507.01714 [quant-ph]} \BibitemShut
  {NoStop}%
\bibitem [{\citenamefont {Baumeler}\ \emph {et~al.}(2014)\citenamefont
  {Baumeler}, \citenamefont {Feix},\ and\ \citenamefont {Wolf}}]{baumeler14}%
  \BibitemOpen
  \bibfield  {author} {\bibinfo {author} {\bibfnamefont {{\"{A}}.}~\bibnamefont
  {Baumeler}}, \bibinfo {author} {\bibfnamefont {A.}~\bibnamefont {Feix}}, \
  and\ \bibinfo {author} {\bibfnamefont {S.}~\bibnamefont {Wolf}},\ }\href
  {\doibase 10.1103/PhysRevA.90.042106} {\bibfield  {journal} {\bibinfo
  {journal} {Phys. Rev. A}\ }\textbf {\bibinfo {volume} {90}},\ \bibinfo
  {pages} {042106} (\bibinfo {year} {2014})},\ \Eprint
  {http://arxiv.org/abs/1403.7333} {arXiv:1403.7333 [quant-ph]} \BibitemShut
  {NoStop}%
\bibitem [{\citenamefont {Bancal}\ \emph {et~al.}(2013)\citenamefont {Bancal},
  \citenamefont {Barrett}, \citenamefont {Gisin},\ and\ \citenamefont
  {Pironio}}]{Bancal:2013aa}%
  \BibitemOpen
  \bibfield  {author} {\bibinfo {author} {\bibfnamefont {J.-D.}\ \bibnamefont
  {Bancal}}, \bibinfo {author} {\bibfnamefont {J.}~\bibnamefont {Barrett}},
  \bibinfo {author} {\bibfnamefont {N.}~\bibnamefont {Gisin}}, \ and\ \bibinfo
  {author} {\bibfnamefont {S.}~\bibnamefont {Pironio}},\ }\href {\doibase
  10.1103/PhysRevA.88.014102} {\bibfield  {journal} {\bibinfo  {journal} {Phys.
  Rev. A}\ }\textbf {\bibinfo {volume} {88}},\ \bibinfo {pages} {014102}
  (\bibinfo {year} {2013})},\ \Eprint {http://arxiv.org/abs/1112.2626}
  {arXiv:1112.2626 [quant-ph]} \BibitemShut {NoStop}%
\bibitem [{\citenamefont {Gallego}\ \emph {et~al.}(2012)\citenamefont
  {Gallego}, \citenamefont {W\"{u}rflinger}, \citenamefont {Ac{\'{i}}n},\ and\
  \citenamefont {Navascu\'{e}s}}]{Gallego:2012aa}%
  \BibitemOpen
  \bibfield  {author} {\bibinfo {author} {\bibfnamefont {R.}~\bibnamefont
  {Gallego}}, \bibinfo {author} {\bibfnamefont {L.~E.}\ \bibnamefont
  {W\"{u}rflinger}}, \bibinfo {author} {\bibfnamefont {A.}~\bibnamefont
  {Ac{\'{i}}n}}, \ and\ \bibinfo {author} {\bibfnamefont {M.}~\bibnamefont
  {Navascu\'{e}s}},\ }\href {\doibase 10.1103/PhysRevLett.109.070401}
  {\bibfield  {journal} {\bibinfo  {journal} {Phys. Rev. Lett.}\ }\textbf
  {\bibinfo {volume} {109}},\ \bibinfo {pages} {070401} (\bibinfo {year}
  {2012})},\ \Eprint {http://arxiv.org/abs/1112.2647} {arXiv:1112.2647
  [quant-ph]} \BibitemShut {NoStop}%
  \bibitem [{SM()}]{SM}%
\BibitemOpen
  \href@noop {} {}\bibinfo {howpublished} {See Supplemental Material in the
  arXiv `ancillary files' for the full list of causal inequalities, the process
  matrices violating them, and further analysis.}\BibitemShut {Stop}%
\end{thebibliography}%

\end{document}